\documentclass[11pt]{article}

\usepackage{array}
\usepackage{amsfonts,amsmath,amssymb,amsthm,mathrsfs}
\usepackage{adjustbox}
\usepackage{tikz}
\usepackage{enumitem}
\usepackage{ifthen}
\usepackage{fullpage}

\usepackage[ruled]{algorithm2e} % For algorithms

\SetAlFnt{\small}
\SetAlCapFnt{\small}
\SetAlCapNameFnt{\small}
\SetAlCapHSkip{0pt}
\IncMargin{-\parindent}

\usepackage{booktabs} % For formal tables
\usepackage{xspace}
\usepackage[pagebackref=true,hidelinks]{hyperref}
\usepackage{accents} % for underbar
\usepackage{pgf} % for pgf figure; could be removed
\usepackage{xparse} % for DeclareDocumentCommand
\usepackage{bbm}
\usepackage{enumitem}
\usepackage[numbers]{natbib}
\usepackage{multirow}
\usepackage{graphicx}
\usepackage{subfigure}
\usepackage{wrapfig}
\usepackage{cleveref}

\graphicspath{ {figures/} }

\usepackage{fixme}
\fxsetup{multiuser,layout={inline,index}}
\fxusetheme{colorsig}
\FXRegisterAuthor{bm}{abm}{BM}
\FXRegisterAuthor{sc}{asc}{SC}
\FXRegisterAuthor{yt}{ayt}{YT}
\FXRegisterAuthor{dp}{adp}{DP}

\newtheorem{theorem}{Theorem}[section]

\newtheorem{lemma}[theorem]{Lemma}
\newtheorem{corollary}[theorem]{Corollary}
\newtheorem{fact}[theorem]{Fact}
\newtheorem{claim}{Claim}

\theoremstyle{definition}
\newtheorem{definition}{Definition}[section]

\newenvironment{proofof}[1]{\vspace{0.1in}\noindent{\em Proof of #1.}}{\hfill\qed}

\newenvironment{numberedtheorem}[1]{%
\renewcommand{\thetheorem}{#1}%
\begin{theorem}}{\end{theorem}\addtocounter{theorem}{-1}}

\newenvironment{numberedlemma}[1]{%
\renewcommand{\thetheorem}{#1}%
\begin{lemma}}{\end{lemma}\addtocounter{theorem}{-1}}

\newcommand{\half}{\frac12}

\newcommand{\eps}{\varepsilon}

\newcommand{\argmax}{\operatorname{arg\,max}}
\newcommand{\argmin}{\operatorname{arg\,min}}

\newcommand{\union}{\cup}
\newcommand{\Union}{\bigcup}
\newcommand{\intersect}{\cap}
\newcommand{\Land}{\mathord{\adjustbox{valign=B,totalheight=.7\baselineskip}{$\bigwedge$}}}

\newcommand{\prob}[2][]{\textrm{\bf Pr}\ifthenelse{\not\equal{}{#1}}{_{#1}}{}\!\left[#2\right]}
\newcommand{\expect}[2][]{\textrm{\bf E}\ifthenelse{\not\equal{}{#1}}{_{#1}}{}\!\left[#2\right]}

\newcommand{\alloc}[1][]{X\ifthenelse{\not\equal{}{#1}}{_{#1}}{}}
\newcommand{\paymt}[1][]{P\ifthenelse{\not\equal{}{#1}}{_{#1}}{}}
\newcommand{\price}[1][]{p\ifthenelse{\not\equal{}{#1}}{_{#1}}{}}
\newcommand{\prices}{\mathbf{p}}
\newcommand{\feas}{{\mathcal{F}}}

\newcommand{\sw}{{\normalfont\textsc{SW}}}
\newcommand{\rev}{{\normalfont\textsc{Rev}}}
\newcommand{\util}{{\normalfont\textsc{Util}}}

\newcommand{\opt}{{\normalfont\textsc{Opt}}}

\newcommand{\costs}{\vec{c}}
\newcommand{\cti}[1][ti]{c_{#1}}

\newcommand{\capt}{B_t}

\newcommand{\alg}{\mathscr{A}}
\newcommand{\coloralg}{\mathscr{B}}

\newcommand{\items}{T}
\newcommand{\tlayer}{J}

\newcommand{\vj}[1][j]{v_{#1}}
\newcommand{\vmax}{v_{\max}}
\newcommand{\vmin}{v_{\min}}

\newcommand{\intj}[1][j]{{I_{#1}}}
\newcommand{\Jvi}{J_{\vali}}
\newcommand{\larmj}[1][j]{{{#1}^1}}
\newcommand{\rarmj}[1][j]{{{#1}^2}}
\newcommand{\armj}[1][j]{{{#1}^a}}

\newcommand{\peakj}[1][j]{{\Lambda_{#1}}}
\newcommand{\contrib}{\gamma}
\newcommand{\phat}{\hat{p}}

\newcommand{\val}{v}
\newcommand{\valj}[1][j]{\val_{#1}}
\newcommand{\vali}[1][i]{\val_{#1}}

\newcommand{\disti}[1][i]{F_{#1}}

\NewDocumentCommand{\soldtk}{ O{t} O{k} }{{\textrm{SOLD}(#1,#2)}}
\NewDocumentCommand{\soldtkj}{ O{t} O{k} }{{\textrm{SOLD}_j(#1,#2)}}

\newcommand{\indicate}{\mathbbm{1}}
\newcommand{\ind}{\mathbbm{1}}

\newcommand{\pricek}{{\price[k]}}

\newcommand{\overbar}[1]{\mkern 1.5mu\overline{\mkern-1.5mu#1\mkern-1.5mu}\mkern 1.5mu}

\newcommand{\fopt}{{\normalfont\textsc{FracOpt}}}

\newcommand{\fval}{{\normalfont\textsc{FracVal}}}
\newcommand{\fwt}{{\normalfont\textsc{FracWt}}}
\newcommand{\Int}{{\mathbbm{I}}}
\newcommand{\heavy}{{\mathcal{H}}}
\newcommand{\light}{{\mathcal{L}}}
\newcommand{\Gtilde}{{\tilde{G}}}
\newcommand{\layer}[1][r]{Y_{#1}}
\newcommand{\Bmax}{{B_{\max}}}
\newcommand{\xr}[1][r]{{\tilde{x}^{(#1)}}}
\newcommand{\xpr}[1][r]{{x'^{(#1)}}}

\newcommand{\qj}{{q_{j}}}
\newcommand{\qvi}[1][i]{{q_{v_{#1}}}}
\newcommand{\xj}[1][j]{{x_{#1}}}

\newcommand{\densj}[1][j]{\rho_{#1}}
\newcommand{\fvt}[1][t]{\textsc{val}_{#1}}
\newcommand{\fvI}[1][I]{\textsc{val}\left(#1\right)}

\newcommand{\lmax}{{\ell_{\max}}}
\newcommand{\offset}{{t_0}}
\newcommand{\fwe}[1][t]{\textsc{w}_{#1}}

\def\Z{\ensuremath{\mathbb{Z}}}

\begin{document}
% Title portion. Note the short title for running heads 
\title{Pricing for Online Resource Allocation: Intervals and
  Paths}
\author{Shuchi Chawla
\thanks{
    {University of Wisconsin-Madison}. \tt{\{shuchi, bmiller, yifengt\}@cs.wisc.edu.}}
\and J. Benjamin Miller$^\dagger$
\and Yifeng Teng$^\dagger$
}
\date{}

\maketitle
\thispagestyle{empty}

\begin{abstract}
We present pricing mechanisms for several online resource allocation
problems which obtain tight or nearly tight approximations to social
welfare. In our settings, buyers arrive online and purchase bundles of
items; buyers' values for the bundles are drawn from known
distributions. This problem is closely related to the so-called
prophet-inequality of \citet{KS-78} and its extensions in recent
literature. Motivated by applications to cloud economics, we consider
two kinds of buyer preferences. In the first, items correspond to
different units of time at which a resource is available; the items
are arranged in a total order and buyers desire {\em intervals} of
items. The second corresponds to bandwidth allocation over a tree
network; the items are edges in the network and buyers desire {\em
paths}.

% We study an online resource allocation problem where a seller has many
% items with multiplicities to offer, and buyers are interested in buying
% bundles of items with values drawn from a known distribution. Our
% goal is to design a truthful online allocation mechanism that
% maximizes social welfare. This problem is closely related to the
% so-called prophet-inequality of Krengel et al. and its extensions in
% recent literature. Motivated by applications to cloud economics, we
% consider two kinds of buyer preferences. In the first, items are
% arranged in a total order, and buyers are interested in buying {\em
% intervals} of items. In the second, the items are the edges in a
% network, and buyers are interested in buying {\em paths}; we focus on
% the special case of tree networks.

Because buyers' preferences have complementarities in the settings we
consider, recent constant-factor approximations via item prices do not
apply, and indeed strong negative results are known. We develop
{\em static, anonymous bundle pricing} mechanisms.
% In the settings we consider buyers' preferences have
% complementarities, so that recently-developed constant-factor

%---the seller partitions his supply into bundles of
%items and determines a price for each bundle. The buyers arrive in
%arbitrary order and purchase their favorite bundles while supply
%lasts.

For the interval preferences setting, we show that static, anonymous
bundle pricings achieve a sublogarithmic competitive ratio, which is
optimal (within constant factors) over the class of all online
allocation algorithms, truthful or not. For the path preferences
setting, we obtain a nearly-tight logarithmic competitive ratio. Both
of these results exhibit an exponential improvement over item pricings
for these settings. Our results extend to settings where the seller
has multiple copies of each item, with the competitive ratio
decreasing linearly with supply. Such a gradual tradeoff between
supply and the competitive ratio for welfare was previously known only
for the single item prophet inequality.

\end{abstract}
\newpage

\setcounter{page}{1}

\section{Introduction}

Consider an online resource allocation setting in which a seller offers
multiple items for sale and buyers with preferences over bundles of items
arrive over time. We desire an incentive-compatible mechanism for allocating
items to buyers that maximizes social welfare---the total value obtained by
all the buyers from the allocation. In the offline setting, this is easily
achieved by the VCG mechanism. In the online context, however, we would like
the mechanism to allocate items and charge payments as buyers arrive, without
waiting for future arrivals.  What would such a mechanism look like, and can it
obtain close to the optimal social welfare?

% In this paper,
We consider two settings for online resource allocation motivated by
applications in cloud economics. In the {\em interval preferences} setting, a
cloud provider has multiple copies of a single resource available to allocate
over time. Customers have jobs that require renting the resource for some
amount of time. If we think of each time unit as a different item,
% in this setting
customers desire {\em intervals} of items. Different intervals,
corresponding to scheduling a job at different times or renting the resource
for different lengths of time, may bring different values to the customer. This
model closely follows the framework described in \cite{babaioff2017era}. Our
second setting, the {\em path preferences setting}, models bandwidth allocation
over a communication network. Each customer is a source-sink pair in the
network that wishes to communicate, and assigns values to paths in the network
between the source and the sink. The items are the edges in the network. We
focus on the special case where the network is a tree.\footnote{Observe that
  the interval preferences setting is a special case of the path preferences
  setting.} See \cite{Jalaparti2016} for further motivation behind this
setting.

It is easy to observe that because of the online nature of the problem
no algorithm for online resource allocation, truthful or otherwise,
can obtain optimal social welfare: the competitive ratio is at least
$2$ even when there is only a single item for sale and the seller
knows both the order of arrival of buyers as well as their value
distributions.\footnote{Suppose there is one item and two buyers. The
  buyer that arrives first has value $1$ for the item; the second
  buyer has value $1/\eps$ with probability $\eps$, and $0$ otherwise,
  for some small $\eps>0$. The optimal allocation is to give the item
  to the first buyer if the second buyer has zero value for it, and
  otherwise give it to the second buyer. This achieves social welfare
  of $2-\eps$ in expectation over the buyers' values. However, any
  online algorithm produces an allocation with expected welfare at
  most~$1$.}  In a remarkable result, \citet{FGL15} showed that this
gap of 2 is the worst possible over a large class of buyer
preferences: a particularly simple and natural incentive-compatible
mechanism, namely {\bf posted item pricing}, achieves a
$2$-approximation to the optimal social welfare in those settings.

Posted pricing is perhaps the most ubiquitous real-world mechanism for
allocating goods to consumers. Supermarkets are a familiar example:
the store determines prices for items, which may be sold individually
or packaged into bundles. Customers arrive in arbitrary order and
purchase the items they most desire at the advertised prices, unless
they're sold out. Many other domains have a similar sequential posted
pricing format, from airfares to online retail to concert tickets.
\citeauthor{FGL15}'s results apply to settings where buyers' values
over bundles of items are fractionally subadditive, a.k.a. XOS. In
these settings, the seller determines a price for each item based on
his knowledge of the buyers' value distributions. These prices
are anonymous and static, meaning that the same prices are offered to
each customer and remain unchanged until supply runs out.

However, the settings we consider exhibit complementarity in buyers' values:
buyers may require certain minimal bundles of items to satisfy their
requirements; anything less brings them zero value. For these settings,
\citet{FGL15} show that anonymous item pricings cannot achieve a competitive
ratio better than linear in the degree of complementarity (that is, the size of
the minimal desired bundle of items). See Footnote~\ref{note:item-pricing-L}
for an example. Is a better competitive ratio possible? Can it be achieved
through a truthful mechanism as simple as static, anonymous item pricings?

\vspace{0.15in}

\parbox[c]{0.95\textwidth}{{\bf We show that near-optimal
    competitive ratios can be achieved for the interval and path
    preferences settings via a static, anonymous bundle pricing
    mechanism. }}

\vspace{0.15in}

Our mechanism is a {\bf posted bundle pricing}: the seller partitions
items into bundles, and prices each bundle based on his knowledge of
the distribution of buyers' preferences. Customers arrive in arbitrary
order and purchase the bundles they most desire at the advertised
prices, unless they're sold out. As in \citeauthor{FGL15}'s work, our
bundle pricings are static and anonymous.

We now elaborate on our results and techniques.

\subsection{Our results}

Recall that in the interval preferences setting, items are arranged in
a total order and buyers desire intervals of items. We assume that
each buyer's value function is drawn independently from some arbitrary
but known distribution over possible value vectors. The seller's
computational problem is essentially a stochastic online interval
packing. Let $L$ denote the length of the longest interval that may be
desired by some buyer. \citet{FGL15} show that no item pricing can
achieve a $o(L)$ competitive ratio in this
setting. \citet{im2011secretary} previously showed that in fact no
online algorithm can achieve a competitive ratio of
$o\left(\frac{\log L}{\log\log L}\right)$. Our first main result matches this
lower bound to within constant factors.

\begin{theorem}
\label{thm:unit-cap-ub}
For the interval preferences setting, there exists a static, anonymous bundle
pricing with competitive ratio $O\left(\frac{\log L}{\log\log
    L}\right)$ for social welfare.
\end{theorem}

The interval preferences setting was studied recently by
\citet{CDH+17}, who also designed a static {\em item} pricing
mechanism for the problem. \citeauthor{CDH+17} showed that when the
item supply is large, specifically, when the seller has
$\tilde\Omega(L^6/\eps^3)$ copies of each item for some $\eps>0$,
static anonymous item pricings achieve a $1-\eps$ approximation to
social welfare. \citeauthor{CDH+17}'s result suggests that there is a
tradeoff between item supply and the performance of posted pricings in
this setting. Our second main result maps out this tradeoff exactly
(to within constant factors) for {\em bundle} pricing. We show that
the approximation factor decreases inversely with item supply, and is
a constant when supply is $\Omega(\log L)$.  In other words, to
achieve a constant approximation via bundle pricing, we require an
exponentially smaller bound on the item supply relative to that
required by \citeauthor{CDH+17} for a $1-\eps$
approximation. Furthermore, this tradeoff is tight to within constant
factors. 

\begin{theorem}
\label{cor:large-cap-ub}
For the interval preferences setting, if every item has at least $B>0$
copies available, then a static, anonymous bundle pricing achieves a
competitive ratio of
\[
    O\left(\frac1B\frac{\log L}{\log\log L - \log B}\right)
\] 
when $B<\log L$, and $O(1)$ otherwise.
\end{theorem}

\begin{theorem}
\label{thm:lb}
For the interval preferences setting, if every item has at least $B>0$
copies available, no online algorithm can obtain a competitive ratio
of $o\left(\frac1B\frac{\log L}{\log\log L}\right)$ for social
welfare.
\end{theorem}

We then turn to the path preferences setting, which appears to be
considerably harder than the special case of interval
preferences. % Recall that here buyers desire paths in a tree
% network. Observe that
In this setting buyers are single-minded in that they desire a
particular path, although their value for this path
is unknown to the seller. As before let $L$ denote the length of the
longest bundle/path that any buyer desires. We first observe that no
online algorithm can achieve a subpolynomial in $L$ competitive ratio
for this setting (see Theorem~\ref{thm:treelowerbound}). We therefore
explore competitive ratios as functions of $H$, the ratio of the
maximum possible value to the minimum possible non-zero value. In
terms of $H$, the construction of \citet{im2011secretary} provides a
lower bound of $\Omega(\log H/\log\log H)$ on the competitive ratio of
any online algorithm. We nearly match this lower bound:

\begin{theorem}
\label{thm:trees-ub}
For the path preferences setting, there exists a static, anonymous bundle
pricing with competitive ratio $O\left(\log H\right)$ for social welfare.
\end{theorem}

%\scnote{Write about improvement with B.}

As for the interval preferences setting, we obtain a linear tradeoff
between item supply and the performance of bundle pricing for the path
preferences setting.
\begin{theorem}
\label{cor:tree-large-cap-ub}
For the path preferences setting, if every edge has at least $B>0$
copies available, then a static, anonymous bundle pricing achieves a
competitive ratio of
\(
    O\left(\frac1B\log H\right)
\)
for welfare.
\end{theorem}
Theorems~\ref{thm:unit-cap-ub} and \ref{cor:large-cap-ub} for the
interval preferences setting are proved in
Section~\ref{sec:ub-interval}. Section~\ref{sec:lowerbound} presents
the lower bounds---Theorems~\ref{thm:lb} and
\ref{thm:treelowerbound}. Our main results for the path preferences
setting, Theorem~\ref{thm:trees-ub} and Theorem~\ref{cor:tree-large-cap-ub}, are proved in
Section~\ref{sec:trees-ub}. Some of our results extend with the same
competitive ratios to settings with non-decreasing marginal production
costs instead of a fixed supply; this setting is discussed in{}
Section~\ref{sec:costs}. Any proofs skipped in the main body of the
paper can be found in Section~\ref{sec:deferred}.

All of our theorems are constructive but require an explicit
description of the buyers' value distributions. The pricings
guaranteed above can be constructed in time polynomial in the sum over
the buyers of the sizes of the supports of their value distributions.

\subsection{Related work and our techniques}

Competitive ratios for welfare in online settings are known as prophet
inequalities following the work of \citet{KS-78} and \citet{Cahn84} on
the special case of allocating a single item. Most arguments for
prophet inequalities follow a standard approach, introduced by
\citet{KW12} and further developed by \citet{FGL15}, of setting
``balanced'' prices or thresholds for each buyer: informally, prices
should be low enough that buyers can afford their optimal allocations,
and at the same time high enough so that allocating items
non-optimally recovers as revenue a good fraction of the optimal
welfare that is ``blocked'' by that allocation. Given such balanced
prices, we can then account for the social welfare of the online
allocation in two components---the seller's share of the welfare,
namely his revenue, and the buyers' share of the welfare, namely their
utility. On the one hand, the prices of any items/bundles sold by the
mechanism contribute to the seller's revenue. On the other hand, any
item/bundle that goes unsold in the online allocation contributes to
the buyers' utility: the buyer who receives it in the optimal
allocation forgoes it in the online allocation for another one with
even higher utility. The argument asserts that in either case good
social welfare is achieved.

When buyers have values with complements and an item pricing is used
this argument breaks down. In particular, it may be the case that a
bundle allocated by the optimal solution goes unsold because only one
of the items in the bundle is sold out.  The loss of the buyer's
utility in this case may not be adequately covered by the revenue
generated by the sold subset of items.\footnote{\label{note:item-pricing-L} Consider, for
  concreteness, the following example due to \citeauthor{FGL15}
  Suppose there are two buyers and $L$ items; the first has value 1
  for any non-empty allocation (i.e., is unit-demand with value 1 for
  every item), and the second has value $L - \eps$ for the set of all
  items and value zero for every subinterval. Observe that for any
  item prices, either the first buyer will be willing to purchase some
  item, thereby blocking the second buyer from purchasing anything, or
  else the second buyer will be unwilling to purchase the full set of
  items. As the optimal welfare is $L-\eps$, no item prices lead to
  better than an $O(L)$-approximation.} We therefore consider pricing
bundles of items.

Recently \citet{DFKL-17} developed a framework for obtaining balanced
prices via a sort of extension theorem. They showed that if one can
define prices achieving a somewhat stronger balance condition in the
full-information setting, where the seller knows the buyers' value
functions exactly, then a good approximation through posted prices can
be obtained in the Bayesian setting as well. A significant benefit of
using this approach is that it suffices to focus on the
full-information setting, and the designer need no longer worry about
the value distribution.

In the full-information setting, even with arbitrary values over
bundles of items, it is easy to construct a static, anonymous bundle
pricing that achieves a $2$-approximation for welfare, as
demonstrated, for example, by \citet{Cohen-Addad-16}.
%\footnote{See Theorem 5.1 in that work.} 
Unfortunately, these prices do not satisfy the strong balance
condition of \citeauthor{DFKL-17} In fact we do not know of any
distribution-independent way of defining bundle prices in the full
information setting that satisfy the balance condition of
\citeauthor{DFKL-17} with approximation factors matching the ones we
achieve. Furthermore, while a main goal of our work is to establish a
tradeoff between supply and competitive ratio for welfare, the
framework of \citeauthor{DFKL-17} does not seem to lend itself to such
a tradeoff.

The crux of our argument for both of our settings lies in constructing a
distribution-dependent partition of items into bundles, and pricing (subsets
of) these bundles. For the interval preferences setting we construct a
partition for which there exists an allocation of bundles to buyers such that
each buyer receives at most one bundle and a good fraction of the social
welfare is achieved. Given such a bundling, we are essentially left with a
``unit-demand'' setting\footnote{A buyer is said to have unit-demand
  preferences if he desires buying at most one item. In our context, buyers may
  desire buying multiple bundles, but the ``optimum'' we compare against
  allocates at most one to each buyer. This allows for the same charging
  arguments that work in the unit-demand case.} for which a prophet inequality
can be constructed using the techniques described above.  We call such an
allocation a {\em unit allocation}. The main technical depth of this result
lies in constructing such a bundling.

In fact, it is straightforward to construct a bundling that leads to an $O(\log
L)$ competitive ratio:\footnote{A detailed discussion of this solution can be
  found in \cite{im2011secretary}.} pick a random power of $2$ between $1$ and
$L$; partition items into intervals of that length starting with a random
offset; and construct an optimal allocation that allocates entire bundles in
the constructed partition. However, our improved $O(\log L/\log\log L)$
approximation requires much more care in partitioning items into bundles of
many different sizes, and requires the partitioning to be done in a
distribution-dependent manner.

% In order to achieve a tight bound on the
% competitive ratio, we need to treat separately items whose welfare
% contributions come from bundles of many different sizes and those
% whose welfare contributions are dominated by bundles of a certain size.

As mentioned earlier, the path preferences setting generalizes
interval preferences, but appears to be much harder. In particular, we
don't know how to obtain a constant competitive ratio even when all
desired paths are of equal length. We show, however, that if all
buyers have equal values and all edges have equal capacity, it becomes
possible to identify up to two most contentious edges on every path,
and the problem behaves like one where every buyer desires only two
items. For this special case, it becomes possible to construct a
pricing using techniques from \cite{FGL15}. Unfortunately, this
argument falls apart when different items have different multiplicity.
In order to deal with multiplicities, we present a different kind of
partitioning of items into bundles or layers and a constrained
allocation that we call a {\em layered allocation}, such that each
layer behaves like a unit-capacity setting. These ideas altogether
lead to an $O(\log H)$ competitive ratio.

As mentioned previously, our approach lends itself to achieving
tradeoffs between item supply and competitive ratio. On the one hand,
when different items are available to different extents, we need to be
careful in constructing a partition into bundles and in some places
this complicates our arguments. On the other hand, large supply allows
us some flexibility in partitioning the instance into multiple smaller
instances, leading to improved competitive ratios. The key to enabling
this partitioning is a composability property of our analysis: suppose
we have multiple disjoint instances of items and buyers, for each of
which in isolation our argument provides a good welfare guarantee;
then running these instances together and allowing buyers
to purchase bundles from any instance provides at least half the sum
of the individual welfares. In the path preferences setting, obtaining
this composability crucially requires buyers to be single-minded.

\paragraph{Discussion and open questions.} Two implications of our
results seem worthy of further study. First, for the settings we
consider as well as those studied previously, posted pricings perform
nearly as well as arbitrary online algorithms, truthful or not. Can
this be formalized into a meta-theorem that holds for broader
contexts? Second, a modest increase in supply brings about significant
improvements in competitive ratio for the settings we study. However,
once the competitive ratio hits a certain constant factor, our
techniques do not provide any further improvement. Can this tradeoff
be extended all the way to a $1+\epsilon$ competitive ratio?

Another natural open problem is to extend our guarantees for the path
preferences setting to general graphs. This appears challenging. Our
arguments rely on a fractional relaxation of the optimal
allocation. General graphs exhibit an integrality gap that is
polynomial in the size of the graph even when all buyers are
single-minded, have the same values (0 or 1 with some probability),
and have identical path lengths; this integrality gap is driven by the
combinatorial structure exhibited by paths in graphs. Indeed this
setting appears to be as hard as the most general setting with no
constraints on buyers' values. It may nevertheless be possible to obtain a
non-trivial competitive ratio relative to a different relaxation of
the offline optimum.

% Our interval preferences model is related to many classic scheduling
% problems generalized by the Job Interval Selection Problem (JISP) (see
% \cite{Spieksma99,COR06} and references therein).  These works focus on
% finding optimal schedules in the absence of incentive constraints, and
% seek to maximize throughput rather than weighted value. Optimal
% solutions are computationally hard to find, but e.g.  \cite{COR06}
% give a polynomial-time 1.582-approximation. Because our buyers arrive
% online and we seek to maximize weighted value, we require different
% techniques and cannot hope for a constant approximation.

% SPMs were first studied for the problem of revenue maximization, where
% computing the optimal mechanism is a computationally hard problem and no simple
% characterizations of optimal mechanisms are known.
\paragraph{Other related work.}
Sequential pricing mechanisms (SPMs) were first studied for the problem of
revenue maximization, where computing the optimal mechanism is a
computationally hard problem and no simple characterizations of optimal
mechanisms are known. A series of works (e.g., \cite{CHK-07, BH-08, CHMS-10,
CMS-10, bilw-focs14, RW-15, CM-16, CZ-17}) showed that in settings where buyers
have subadditive values, SPMs achieve constant-factor approximations to
revenue. In most interesting settings, good approximations to revenue
necessarily require non-anonymous pricings. As a result, techniques in this
literature are quite different from those for welfare.  \cite{GHKSV14} gives
(non-truthful) online algorithms which obtain constant-factor approximations in
our settings when supply is unit-capacity, buyers are single-minded, and all
values are identical.  There is also a long line of work on welfare-maximizing
mechanisms with buyers arriving online in the worst case setting where the
seller has no prior information about buyers' values \citep{lavi2015online,
hajiaghayi2005online, cole2008prompt, azar2011prompt, azar2015truthful,
CDH+17}. The worst case setting generally exhibits very different techniques
and results relative to the Bayesian setting we study.

\section{Model and definitions}
\label{sec:prelim}

We consider a setting with $n$ buyers and a set of items $\items$. We
index buyers by $i$ and items by $t$ (or $e$ for edges). Buyer $i$'s
valuation function is denoted $\vali: 2^{\items} \rightarrow \Re^+$,
with $\vali(\emptyset)=0$. Our setting is Bayesian: $\vali$ is drawn
from a distribution $\disti[i]$ that is independent of other buyers
and known to the seller. We emphasize that values may be correlated
across bundles, but not across buyers. Let $B_t$ denote the number of
copies available, a.k.a.  supply or capacity, for item $t$. Let $B =
\min_t B_t$. The {\em unit-capacity} setting is a special case where
$B=1$. An allocation is an assignment of bundles of items to buyers
such that no item is allocated more than its number of copies
available. Our goal is to maximize the buyers' total welfare---that
is, the sum over values each buyer derives from his allocated bundle.

\paragraph{Jobs.} We now describe assumptions and notational
shortcuts for buyers' value functions that hold without loss of
generality and simplify exposition. First, we assume that values are
monotone: for all $i$ and $\vali$, and all bundles
$S\subset S'\subseteq \items$, $\vali(S)\le\vali(S')$. Second, we assume
that each buyer's value distribution is a discrete
distribution\footnote{For constructive versions of our results, we
  require the supports of the distributions to be explicitly given,
  however, our arguments about the existence of a good pricing work
  also for continuous distributions.} with 
  finite support. In other words, each buyer $i$ can only have
  finitely many possible valuation functions: with probability $\qvi$,
  buyer $i$'s values are given by the known function $\vali$.
  %a support of size $2$, with
  %one of the two value functions in the support being the all-zeros
  %function.\footnote{This can be enforced by ``splitting'' each buyer
  %into multiple buyers, one for each value function in the support of
  %the original buyer. The arrivals of these new buyers are negatively
  %correlated. It is easy to observe that revenue and welfare
  %guarantees for posted pricings continue to hold under negatively
  %correlated arrivals.} In other words, with some probability $\qvi$,
  %buyer $i$'s values are given by the fixed and known function $\vali$,
  %and with probability $1-\qvi$, his value is $0$ for every bundle of
  %items. When the non-zero value function is instantiated, we say that
  %the buyer has arrived.

Third, we think of a buyer $i$ with value function $\vali$
as being a collection of ``jobs'' represented by tuples $(i,\vali,S)$,
one for each bundle $S$ of items desired by the buyer. Informally,
each job is a potential (minimal) allocation to a buyer with a given
value.\footnote{In particular we remove ``duplicates'', or bundles $S$
  such that for some $S'\subsetneq S$, $\vali(S')=\vali(S)$.} Let
$\Jvi = \{(i,\vali, S)\}$ denote the set of all such jobs for a buyer
$i$ with value function $\vali$; let $U$ denote the union of these
sets over all possible buyers and value functions. When a buyer $i$
with value $\vali$ arrives, we interpret this event as the
simultaneous arrival of all of the jobs in $\Jvi$. If the buyer is
allocated the bundle $S$, we say that job $(i,\vali,S)$ is
allocated. In what follows, it will be convenient to consider jobs as
fundamental entities. We therefore identify a job $(i,\vali,S)$ by the
single index $j$.  Then $\intj$ is the corresponding bundle of items,
$\vj$ is $\vali(I_j)$, and $\qj$ is $\qvi$.

\paragraph{Interval and path preferences.} In the interval preferences
setting, the set $\items$ of items is totally ordered, and buyers
assign values to intervals of items. In particular, for any bundle $S$
of items,
$\vali(S) = \max_{\text{intervals } I \subseteq S} \, \vali(I)$ where
$I$ ranges over all contiguous intervals contained in
$\items$. Accordingly, jobs as defined above also correspond to
intervals. In the path preferences setting, $\items$ corresponds to
the set of edges in a given tree. Each buyer $i$ has a fixed and known
path, denoted $P_i$, and a scalar (random) value $\vali$, with
$\vali(S) = \vali$ for $S\supset P_i$ and $0$ otherwise. Accordingly,
buyers are {\em single-minded} and each instantiated value function of
a buyer is associated with a single job.

\paragraph{Bundle pricings.} The mechanisms we study are static and
anonymous bundle pricings. Let $\prices$ denote such a pricing
function. For the interval preferences setting, we partition the
multiset of items into disjoint intervals and price each interval in
the partition. For the path preferences setting, we partition items
into disjoint ``layers'' and construct a different pricing function for
each layer, which assigns a price to every path contained in that
layer.\footnote{The pricing is essentially a layer-specific
  item pricing, with bundle totals subject to a layer-specific reserve
  price.} Observe that different copies of an item end up in different
bundles/layers, and may therefore be priced differently.  Buyers
arrive in adversarial order. When buyer $i$ arrives, he selects a
subset of remaining unsold bundles to maximize his value for the items
contained in the subset minus the total payment as specified by the
pricing $\prices$.

\vspace{0.1in}
\noindent
Let $\opt$ denote the hindsight/offline optimal expected social welfare, and
$\sw(\prices)$ denote the expected social welfare obtained by the
static, anonymous bundle pricing $\prices$.

% \subsection{Fractional allocations}

% Our analysis makes use of {\em fractional allocations}. A fractional allocation
% $x$ assigns an $\xjki$-fraction of each item in $\intji$ to buyer $j$ when
% her value is $\valjk$.\footnote{Note that fractional allocations are used only
% in determining prices and in bounding the expected optimal (integer) solution,
% not in the mechanism itself.}

% An allocation $x$ is a feasible fractional allocation if it satisfies the
% following constraints.
% \begin{align*}
%     %\max \quad & \sum_j\sum_{k=1}^{\suppj}\qjk\sum_{i=1}^{\intsj}
%             %\valjki\xjki \\
%     \sum_j \sum_{k=1}^\suppj \qjk \sum_{i:\intji\ni t}\xjki & \le B_t
%             \quad \forall t \;\textrm{(supply)} \\
%          \|\xjk\|_1 & \le 1 \quad \forall j, k \;\textrm{(demand)} \\
%          x & \ge 0
% \end{align*}

\paragraph{A fractional relaxation.}
A {\em fractional allocation} $x$ assigns an $\xj$-fraction of each item in
$\intj$ to job $j$.\footnote{Note that fractional allocations are used only in
determining prices and in bounding the expected optimal (integer) solution, not
in the mechanism itself.} A fractional allocation is feasible if it satisfies
the supply and demand constraints of the integral problem: no item may be
allocated more than $\capt$ times and no buyer obtains additional value for more
than one bundle. Let $\feas$ denote the polytope defined by constraints 
\eqref{eq:supply}--\eqref{eq:nonneg} below. 
%Then $x$ is a feasible fractional 
%allocation if and only if $x \in \feas$. 
\begin{align}
    \hfill \sum_{j : \intj \ni t} \xj &\le \capt \quad \forall t
            & & \textrm{(supply)} \label{eq:supply} \\
    \sum_{j \in \Jvi} \xj & \le \qvi \quad \forall \vali
            & & \textrm{(demand)} \label{eq:demand} \\
    x & \ge 0 & & \label{eq:nonneg}
\end{align}
For any $x$, $\fval(x)$ is the total fractional value in $x$ (without
regard for feasibility), and $\fwt(x)$ is the total fractional weight
of $x$:
\[\fval(x) = \sum_j v_j x_j \quad \text{and} \quad \fwt(x) = \sum_j
  x_j.\] 
\noindent
For a subset $A$ of jobs, we use $x_A$ (and sometimes $(x,A)$) to
denote the fractional allocation confined to set $A$ and zeroed out
everywhere else. That is, $(x_A)_j = x_j$ for $j\in A$ and $(x_A)_j=0$
for $j\not\in A$.
\[\fval(x_A) = \sum_{j\in A} v_j x_j\quad \text{and} \quad \fwt(x_A) = \sum_{j\in A} x_j.\]  
\noindent
Fractional allocations provide an upper bound on the optimal welfare;
see \Cref{sec:deferred} for the proof of a more general statement
(Lemma~\ref{lem:fopt-upperbound-costs}). Here
$\fopt := \max_{x \in \feas} \fval(x)$.
% The following lemma is a special case of
% Lemma~\ref{lem:fopt-upperbound-costs};
\begin{lemma}
    \label{lem:fopt-upperbound}
    $\fopt \ge \opt$.
\end{lemma}

\section{Bundle Pricing for Interval Preferences}
\label{sec:ub-interval}

In this section we present our main results,
Theorems~\ref{thm:unit-cap-ub} and \ref{cor:large-cap-ub}, for the
interval preferences setting. We begin by defining a special kind of
fractional allocations that we call {\em fractional unit
  allocations}. These have properties that guarantee the existence of
a bundle-pricing mechanism which obtains a good fraction of the
welfare of the fractional allocation. The intent is to decompose the
fractional allocation across disjoint bundles such that the fractional
value assigned to each bundle can be recovered by pricing that bundle
individually. Section~\ref{sec:frac-unit-alloc} presents a definition of
unit allocations and their connection to pricing.

% Recall that Lemma~\ref{lem:FGL} shows
% that in order to obtain a good approximation ratio via bundle
% pricings, it suffices to design a good fractional unit allocation. 

The remaining technical content of this section then focuses on
designing fractional unit allocations. We begin with the special case
of unit-capacities in Section~\ref{sec:unit-cap-ub}, where we show the
existence of an $O(\log L/\log\log L)$-approximate fractional unit
allocation. In Section~\ref{sec:arbit-cap-ub} we extend our analysis
to the general case of arbitrary multiplicities, proving
Theorem~\ref{thm:unit-cap-ub}. Finally, in
Section~\ref{sec:large-cap-ub} we show that if the capacity for every
item is large enough, specifically at least $B$, then the
approximation ratio decreases by a factor of $\Theta(B)$
(Theorem~\ref{cor:large-cap-ub}).

\subsection{Fractional unit allocations}
\label{sec:frac-unit-alloc}

\begin{definition} A fractional allocation $x$ is a {\em fractional
    unit allocation} if there exists a partition of the multiset of
  items (where item $t$ has multiplicity $\capt$) into bundles
  $\{T_1, T_2, T_3, \cdots\}$, and a corresponding partition of jobs
  $j\in U$ with $x_j>0$ into sets $\{A_1, A_2, A_3, \cdots\}$, such
  that:
    \begin{itemize}
        \item For all $j\in U$ with $x_j>0$, there is exactly one
          index $k$ with $j\in A_k$. 
        \item For all $k$ and $j\in A_k$, $I_j\subseteq T_k$.
        \item For all $k$, we have $\fwt(x_{A_k}) \le 1$.
   \end{itemize}
\end{definition}

A note on the terminology: we call fractional allocations satisfying
the above definition {\em unit} allocations because once the partition
of items is specified, each job can be assigned at most one bundle in
the partition and each bundle can be fractionally assigned to at most
one job. % We emphasize that each item $t$ may appear in at most $\capt$ different sets $T_k$, and
% each job with nonzero fractional allocation is associated with exactly one
% set $T_k$.

We note that given any fractional unit allocation $x$, for any
instantiation of values, it is possible to define a pricing function
over the bundles $T_k$ that is $(1,1)$-balanced with respect to $x$
within the framework of \cite{DFKL-17}. This is because fractional unit
allocations behave essentially like feasible allocations for
unit-demand buyers. For completeness we present a simpler
first-principles argument based on the techniques of \citet{FGL15}
showing that such a pricing obtains at least half the value of the
fractional unit allocation. The proof is deferred to
Section~\ref{sec:deferred}.

\begin{lemma}
    \label{lem:FGL}
    For any feasible fractional unit allocation $x$, there exists a static,
    anonymous bundle pricing $\prices$ such that
    \[
        \sw(\prices) \ge \half \fval(x).
    \]
\end{lemma}

    \subsection{The unit capacity setting}
\label{sec:unit-cap-ub}

% We begin with some definitions. (Move some of these to section 2?) Let $x$ denote a fractional allocation. For a subset $A$ of jobs, we use $x_A$ to denote the fractional allocation confined to set $A$ and zeroed out everywhere else. That is, $(x_A)_j = x_j$ for $j\in A$ and $=0$ for $j\not\in A$.
% $\fval(x)$ is the total fractional value in $x$ without worrying about feasiblity: $\fval(x) = \sum_j v_j x_j$ and $\fval(x_A) = \sum_{j\in A} v_j x_j$. 
% $\fwt(x)$ is the total fractional weight of $x$: $\fwt(x) = \sum_j x_j$ and $\fwt(x_A) = \sum_{j\in A} x_j$. 
% 
% We say that a fractional allocation $x$ is a {\em fractional unit allocation} if there exists a partition of the items into intervals $\{T_1, T_2, T_3, \cdots\}$ such that:
% \begin{itemize}
% \item For all $j\in U$ with $x_j>0$, there exists an index $i$ with $I_j\subseteq T_i$ and $I_j\cap T_{i'}=\emptyset$ for all $i'\ne i$.
% \item Denoting $A_i = \{j: I_j\subseteq T_i\}$, we have $\fwt(x_{A_i})\le 1$.
% \end{itemize}

The main claim of this section is that for any feasible fractional allocation
$x$ there exists a fractional unit allocation $x'$ such that $\fval(x)\le
O(\beta)\fval(x')$ where $\beta$ is defined such that $\beta\log\beta = \log
L$. Then, if $x$ is the optimal fractional allocation for the given instance,
we can apply Lemma~\ref{lem:FGL} to the allocation $x'$ to obtain a pricing
that gives an $O(\beta)$ approximation. Observe that $\beta=O\big(\frac{\log
L}{\log\log L}\big)$, so we get the desired approximation factor. A formal
statement of this claim is given at the end of this subsection.

Before we prove the claim, let us discuss the intuition behind our analysis.
One approach for producing a fractional unit allocation is to find an
appropriate partition of items into disjoint intervals $\{T_1, T_2, \cdots\}$,
and remove from $x$ all jobs that do not fit neatly into one of the intervals.
We can then rescale the fractional allocations of jobs within each interval
$T_i$, so that their total fractional weight is $1$. The challenge in carrying
out this approach is that if the intervals $T_i$ are too short, then they leave
out too much value in the form of long jobs. On the other hand, if they are too
long, then we may require a large renormalizing factor for the weight, again
reducing the value significantly. 

To account for these losses in a principled manner we consider a suite
of nested partitions, one at each length scale, of which there are
$\log L$ in number.\footnote{All of our arguments extend trivially to
  settings where $L$ denotes the ratio of lengths of the largest to the smallest
  interval of interest, because the number of relevant length scales
  is logarithmic in this ratio.} We then place a job of length $\ell$ in the interval that contains it (if one exists) at length scale $\sim 2\ell$. The intervals over all of the length scales together capture much of the fractional value of the allocation $x$. Furthermore, the fractional weight within each interval can be bounded by a constant. At this point, any single partition gives us a fractional unit allocation. However, since the number of length scales is $\log L$, picking a single one of these unit allocations only gives an $O(\log L)$ approximation in the worst case. In order to do better, we argue that there are two possibilities: either (1) it is the case that many intervals have a much larger than average contribution to total weight, in which case grouping these together provides a good unit allocation; or, (2) it is the case that most intervals have low total weight, so that we can put together multiple length scales to obtain a unit allocation without incurring a large renormalization factor for weight.

We now proceed to prove the claim. The proof consists of modifying $x$ to obtain $x'$ through a series of refining steps. Steps 1 and 2 define the suite of partitions and placement of jobs into intervals as discussed above. Step 3 provides a criterion for distinguishing between the cases (1) and (2) above. Step 4 then provides an analysis of case (1), and Step 5 an analysis of case (2).

\subsubsection*{Step 1: Filtering low value jobs}

We first filter out jobs that do not contribute enough to the solution $x$ depending on when they are scheduled. Accordingly, we define:
\begin{itemize}
    \item For each job $j$, the {\em value density} of $j$ is $\densj = v_j/|I_j|$.
    \item For each item $t$, the {\em fractional value at $t$} is $\fvt(x) =
      \sum_{j : I_j \ni t} \densj x_j$.  Note $\fval(x) = \sum_t\fvt(x)$. We drop the argument $x$ when it is clear from the context.
    \item For any set of items $I$ (that may or may not be an interval), define $\fvI = \sum_{t \in I}\fvt$.
\end{itemize}

In this step, we remove from consideration jobs with fractional value less than half the total fractional value of their interval. In particular, let $U_1 = \{j : \valj\ge\half \fvI[\intj] \}$. The following lemma shows that we do not lose too much value in doing so.
\begin{lemma}
    \label{lem:small-v-bound}
%    Let $J = \{j : \valj \ge \half\fvI[\intj](x)\}$. Then $\fval(x_J) \ge \half \fval(x)$.
    For any fractional allocation $x$ and the set of jobs $U_1$ as
    defined above, we have \[\fval(x_{U_1})\ge \half\fval(x).\]
\end{lemma}
\noindent
The proof of Lemma~\ref{lem:small-v-bound} is a simple counting
argument and we defer it to \Cref{sec:deferred}.

\subsubsection*{Step 2: Bucketing}
\label{sec:bucketing}

In the rest of the argument, we will group jobs according to both their value
and the length of their interval. Let $\ell\in\{1, \cdots, \lceil\log
L\rceil+1\}$ denote a length scale, and let $a\in\{1, \cdots, \lceil \log \vmax
\rceil\}$ denote a value scale. We will further partition jobs of similar
length across non-overlapping intervals. Accordingly, let $\offset$ be an
offset picked u.a.r. from $[2^{\lmax}]$ where $\lmax = \lceil\log L\rceil+1$.
The partition corresponding to length scale $\ell$ then consists of
length-$2^{\ell}$ intervals $\{\Int_{\ell,k}\}_{k\in\Z}$ where
\[
    \Int_{\ell,k}:=[\offset+k 2^{\ell}+1, \offset+(k+1)2^{\ell}].
\]
See \Cref{fig:bucketing}.

We are now ready to define job groups formally. Group $G_{\ell, k,a}$ consists of every job $j$ of length between $2^{\ell-2}$ and $2^{\ell-1}$ and value between $2^{a-1}$ and $2^a$ whose interval lies within the $k$th interval at length scale $\ell$: $I_j\subseteq \Int_{\ell,k}$. We say that the jobs in $G_{\ell, k,a}$ are assigned to interval $\Int_{\ell,k}$.
\[
    G_{\ell, k,a} := \{j: \lceil \log |I_j| \rceil = \ell-1 \quad\Land\quad
            \lceil \log v_j \rceil = a \quad\Land\quad
            I_j\subseteq \Int_{\ell,k} \}
\]

\noindent
$G_{\ell, k}$ denotes all the jobs ``assigned'' to $\Int_{\ell,k}$: $G_{\ell,k}
= \Union_a G_{\ell,k,a}$. 

Observe that our choice of the offset $\offset$ may cause us to drop some jobs,
in particular those that do not fit neatly into one of the intervals in the
partition corresponding to the relevant length scale. Let $U_2 =
\Union_{\ell,k,a} G_{\ell,k,a}$. The next lemma bounds the loss in value at
this step.

\begin{lemma}
    \label{lem:hierarchy-alignment}
    For any fractional allocation $x$ and with $U_1$ and $U_2$ defined as above, $\fval(x_{U_2})\ge \half \fval(x_{U_1})$.
\end{lemma}
\begin{proof}
  Recall that $\offset$ is chosen u.a.r. from $2^\lmax$. Moreover, the length of any job $j$ is at most half the length of the intervals at the scale at which $j$ is considered for bucketing. Therefore, $j$ survives with probability at least $1/2$.
\end{proof}

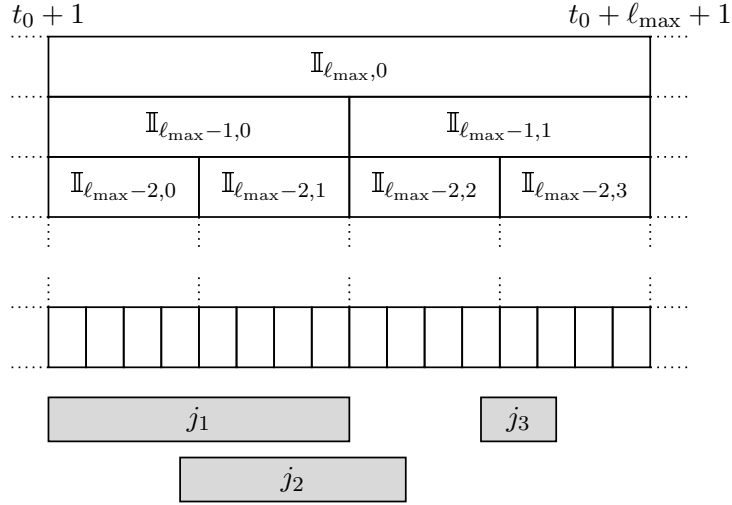
\begin{figure}
    \centering
    \newcommand{\cellwidth}{0.5 cm}
    \newcommand{\cellheight}{0.8 cm}

    \begin{tikzpicture}
        [x=\cellwidth,y=-\cellheight,node distance=0 cm,outer sep=0 pt,line
        width=0.66pt]

        \tikzstyle{cell}=[rectangle,draw,
            minimum height=\cellheight,
            anchor=north west,
            text centered]
        \tikzstyle{w16}=[cell,minimum width=16*\cellwidth]
        \tikzstyle{w8}=[cell,minimum width=8*\cellwidth]
        \tikzstyle{w4}=[cell,minimum width=4*\cellwidth]
        \tikzstyle{w2}=[cell,minimum width=2*\cellwidth]
        \tikzstyle{w1}=[cell,minimum width=1*\cellwidth]
        \tikzstyle{job}=[rectangle,draw,fill=gray!30,anchor=north west]

        % top rows
        \node[anchor=south] at (0, 0) {$\offset + 1$};
        \node[anchor=south] at (16,0) {$\offset + \lmax + 1$};
        \node[w16] at (0, 0) {$\Int_{\lmax,0}$};
        \node[w8]  at (0, 1) {$\Int_{\lmax-1,0}$};
        \node[w8]  at (8, 1) {$\Int_{\lmax-1,1}$};
        \foreach \k in {0,1,2,3} {
            \node[w4] at (4*\k, 2) {$\Int_{\lmax-2,\k}$};
        }

        % horizontal continuation...
        \foreach \y in {0,1,2,3, 4.5,5.5} {
            \draw[dotted] (0, \y) -- (-1,\y);
            \draw[dotted] (16,\y) -- (17,\y);
        }

        % gap
        \foreach \x in {0,4,...,16} {
            \draw[dotted] (\x,3) -- (\x,3.5);
            \draw[dotted] (\x,4.5) -- (\x,4);
        }

        % bottom row
        \foreach \x in {0,...,15} {
            \node[w1] at (\x, 4.5) {};
        }

        % jobs
        \node[job,minimum width=8*\cellwidth] at (   0,6) {$j_1$};
        \node[job,minimum width=6*\cellwidth] at ( 3.5,7) {$j_2$};
        \node[job,minimum width=2*\cellwidth] at (11.5,6) {$j_3$};
    \end{tikzpicture}
    \caption{\small{Bucketing jobs by length as in \hyperref[sec:bucketing]{Step 2}.
    Jobs $j_1$ and $j_2$ have length scale $2^\lmax$ (i.e., lengths between
    $2^{\lmax-2}$, exclusive, and $2^{\lmax-1}$, inclusive) and will therefore
    be assigned to bucket $\Int_{\lmax,0}$. Job $j_3$, however, will be
    dropped; it has length scale $\lmax-2$, but does not fit entirely within
    any bucket at that scale.}}
    \label{fig:bucketing}
\end{figure}

\subsubsection*{Step 3: Classification into heavy-weight and light-weight sets of jobs}

We first discuss the intuition behind Steps 3 and 4. Consider a single
group $G_{\ell,k,a}$. Let us assume briefly, for simplicity, that all
jobs in this group have the same value $v$. Recall that all of the
jobs $j\in G_{\ell,k,a}$ by virtue of being in $U_1$ satisfy the
property that $\fvI[\intj]\le 2\valj=2v$. Furthermore, since every job
$j\in G_{\ell,k,a}$ is of size at least $2^{\ell-2}$ whereas the
interval covered by $G_{\ell,k,a}$, namely $\Int_{\ell,k}$, is of size
$2^\ell$, the total fractional value of $\Int_{\ell,k}$ is\footnote{To
  be precise, the total value of $\Union_{j\in G_{\ell,k,a}} \intj$ in
  this case is at most $8v$.} comparable to $v$.  Therefore, as long
as the jobs in the group $G_{\ell,k,a}$ have sufficient total
fractional weight, this group of jobs alone would recover (a constant
fraction of) the fractional value of the interval $\Int_{\ell,k}$. In
other words, we could then immediately throw away jobs in other groups
(corresponding to other length or value scales) that overlap with this
interval. Step 3 filters out such ``heavy-weight'' groups of jobs and
in Step 4 we argue that these groups form a good unit
allocation. Intervals that do not have any length scales with
heavy-weight groups of jobs are relegated to Step~5.

We now present the details of this argument. First, for every group
$G_{\ell,k,a}$, we break the group up into contiguous
components. Observe, in particular, that the set of items covered by
jobs in $G_{\ell,k,a}$, namely $\Union_{j\in G_{\ell,k,a}} \intj$, may
not be an interval itself but is composed of at most three disjoint
intervals because each contiguous component has length at least a
quarter of $|\Int_{\ell,k}|$. Consider each of the at most three
corresponding sets of overlapping jobs. We will use
$G_{\ell,k,a}^{(i)}$ for $i\in \{1,2,3\}$ to denote these sets. For
any such set $G$ of overlapping jobs, let
$\Int_G = \Union_{j\in G} \intj$ denote the interval of items it
covers.

Next, we classify these sets of jobs into heavy-weight or light-weight. In the following definition, we hide the arguments $\ell,k,a,i$ to simplify notation.
\[\heavy := \{G : \fwt(x_G) \ge 1/6\beta\} \quad \quad \light := \{G : \fwt(x_G) < 1/6\beta\} \]

% For any group $G$ (with the arguments being implicit), if we have $\fwt(x_G) \ge 1/6\beta$, call the set $G$ heavy-weight; let $\heavy$ be the collection of all such sets. Otherwise, call the set light-weight, and let $\light$ be the collection of all such sets. 
\noindent
The following fact is immediate:

\begin{fact}
  \label{fact:heavy-light-classification}
%  $\fvI[\cup_{G\in\heavy}\Int_G] + \sum_{G\in\light} \fval(x_G) \ge $
  $\sum_{G\in\heavy}\fval(x_G) + \sum_{G\in\light} \fval(x_G) \ge \fval(x_{U_2})$. 
%  $\fvI[\cup_{G\in\heavy}\Int_G] + \fvI[\cup_{G\in\light}\Int_G]  \ge \fval(x_{U_2})$. 
\end{fact}

We will now proceed to construct two fractional unit allocations, one of which
extracts the value of the heavy-weight sets (see Step 4), and the other
extracts the value of the light-weight sets (see Step~5).

\subsubsection*{Step 4: Extracting the value of heavy-weight intervals}

Consider any heavy-weight set $G\in\heavy$ (with the arguments $\ell$, $a$,
etc. being implicit). The following lemma is implicit in the discussion above
and follows from the observations that (1) all jobs in $G$ are high-value jobs
(that is, they belong to $U_1$); (2) they have roughly the same value (within a
factor of $2$); (3) the interval $\Int_G$ can be covered by at most $6$ such
jobs; and (4) the total weight of $G$, $\fwt(x_G)$, is at most $1/6\beta$ by
the definition of heavy-weight sets. 
\begin{lemma}
    \label{lem:heavy-Gs}
    For all $j\in G$, $\valj \ge \frac 1{12} \fvI[\Int_G]$, and therefore
    $\fval(x_G)\ge \frac 1{72\beta} \fvI[\Int_G]$.
\end{lemma}

There are two remaining issues in going from the allocation $x_{\Union_\heavy
G}$ to a unit-allocation. First, the intervals $\{\Int_G\}_{G\in\heavy}$
overlap, and second, the fractional weight of $G$ can be larger than $1$. The
second issue can be dealt with by rescaling $x$ by an appropriate factor. To
deal with the first, we state the following simple lemma without proof.
\begin{lemma}
    \label{lem:interval-cover}
    For any given collection $\mathcal C$ of intervals, one can efficiently
    construct two sets $S_1, S_2\subseteq \mathcal C$ such that 
    \begin{itemize}
        \item $S_1$ (and likewise $S_2$) is composed of disjoint intervals;
            that is, $I\cap I' = \emptyset$ for all $I\ne I'\in S_1$.
        \item Together they cover the entire collection $\mathcal C$; that is,
            $\Union_{I\in S_1\union S_2} I = \Union_{I\in\mathcal C} I$.
    \end{itemize}
\end{lemma}

We can now put these lemmas together to construct a unit allocation that covers the fractional value of heavy-weight intervals.

\begin{lemma}
  \label{lem:heavy-bound}
  There exists a fractional unit allocation $\tilde{x}_\heavy$ such that 
\[\sum_{G\in\heavy}\fval(x_G)\le O(\beta)\fval(\tilde{x}_\heavy).\]
\end{lemma}
\begin{proof}
%\scnote{Defer}
    Apply Lemma~\ref{lem:interval-cover} to the collection
    $\{\Int_G\}_{G\in\heavy}$ to obtain sets $S_1$ and $S_2$. We think of $S_1$
    and $S_2$ as sets of groups in $\heavy$. Assume without loss of generality
    that $S_1$ has larger fractional value than $S_2$. Let $\tilde{x}_\heavy$
    be the fractional allocation $x_{\Union_{G\in S_1} G}$ scaled down by a
    factor of $4$. Then we have:
%, it is immediate from Lemmas~\ref{lem:heavy-Gs} and
%    \ref{lem:interval-cover} that 
\begin{align*}
    \fval(\tilde{x}_\heavy) = \frac 14 \fval(x, \union_{G\in S_1} G) & \ge
            \frac 1{O(\beta)} \sum_{G\in S_1} \fvI[\Int_G] \\
    & \ge \frac 1{O(\beta)} \fvI[\union_{G\in\heavy}\Int_G]\ge \frac 1{O(\beta)} \sum_{G\in\heavy}\fval(x_G).
\end{align*}
Here the first inequality follows from applying \Cref{lem:heavy-Gs} to
every $G\in S_1$, and the fact that intervals in $S_1$ are
disjoint. The second inequality follows by recalling that the
intervals in $S_1$ and $S_2$ together cover the
collection$\{\Int_G\}_{G\in\heavy}$, and so,
$\fvI[\union_{G\in S_1}\Int_G]\ge
\half\fvI[\union_{G\in\heavy}\Int_G]$.

It remains to show that $\tilde{x}_\heavy$ is a unit allocation. To
see this, consider the partition of jobs into groups $G$ in $S_1$.
The corresponding collection of intervals
$\{\Int_G\}_{G\in S_1}$ forms a partition of the items by virtue of
the fact that the intervals corresponding to groups in $S_1$ are
disjoint. It remains to argue that $\fwt(\tilde{x}_\heavy, G)\le 1$,
that is, $\fwt(x_G)\le 4$ for all $G\in\heavy$. 

To prove this claim, observe that since each job in $G$ has length
at least a quarter of the length of $\Int_G$, we can find up to four
items in $\Int_G$ such that each job in $G$ contains at least one of
the four items in its interval. Since $x$ is a feasible fractional
allocation, the total weight of all jobs containing any one of those
items is at most $1$, and therefore the total weight of jobs in
$G$ altogether is at most~$4$.
\end{proof}

\subsubsection*{Step 5: Extracting the value of light-weight intervals}

We now consider the light-weight groups $G_{\ell,k,a}^{(i)}$. As discussed at the beginning of this section, in order to obtain a good approximation from these sets, we must construct a unit allocation out of jobs at multiple length scales.

Define $\Gtilde_{\ell,k,a} = \Union \{G_{\ell,k,a}^{(i)}\in\light\}$
to be the set of all jobs in $\Int_{\ell,k}$ with value scale $a$ that
belong to light-weight groups. Let $\Gtilde_{\ell,k}=\Union_a \Gtilde_{\ell,k,a}$
denote all light-weight jobs assigned to $\Int_{\ell,k}$.
Since each individual group
$G_{\ell,k,a}^{(i)}$ in $\light$ has total weight at most $1/6\beta$,
we have that $\fwt(\Gtilde_{\ell,k,a})\le 1/2\beta$. In order to
obtain a partition of jobs and items, we would now like to associate a
single set of jobs with each partition $\Int_{\ell,k}$ that is
simultaneously high value and low weight. Unfortunately, the set
$\Gtilde_{\ell,k}$ may have very large total weight since
it combines together many low weight sets. We use the fact that the
values of jobs in these low weight sets increase geometrically to
argue that it is possible to extract a subset of jobs from
$\Gtilde_{\ell,k}$ that is both light-weight (i.e. has
total weight at most $1/\beta$) and captures a large fraction of the
total value in $\Gtilde_{\ell,k}$.

\begin{lemma}
    \label{lem:a-cell-bound}
    For each interval $\Int_{\ell,k}$, there exists a set of jobs
    $S_{\ell,k} \subseteq \Gtilde_{\ell,k}$ such that
    \begin{enumerate}[label=\roman*., leftmargin=2\parindent]
        \item $\fwt(x, S_{\ell,k}) \le \frac1\beta$ and
        \item $\fval(x, S_{\ell,k}) \ge \frac16 \fval(x, \Gtilde_{\ell,k})$.
    \end{enumerate}
\end{lemma}

We defer the proof of Lemma~\ref{lem:a-cell-bound} to
\Cref{sec:deferred}. The remainder of our analysis then hinges on the fact that if we consider the partition into intervals at some length scale $\ell$, namely $\{\Int_{\ell,k}\}_{k\in\Z}$, and consider for every interval in this partition the set of all jobs in this interval at the $\log \beta$ length scales below $\ell$, the total fractional weight of these jobs is at most $1$. We therefore obtain a unit allocation while capturing the fractional value in $\log\beta$ consecutive scales.

\begin{lemma}
    \label{lem:light-bound}
    There exists a fractional unit allocation $\tilde{x}_\light$ such that 
    \[
        \sum_{G\in\light} \fval(x_G) \le
        O\left(\frac{\log L}{\log\beta}\right)\fval(\tilde{x}_\light).
    \]
\end{lemma}

We defer the proof of Lemma~\ref{lem:light-bound} to \Cref{sec:deferred}.

\subsubsection*{Putting everything together}

Combining Lemmas~\ref{lem:small-v-bound} and \ref{lem:hierarchy-alignment},
Fact~\ref{fact:heavy-light-classification}, and Lemmas~\ref{lem:heavy-bound}
and \ref{lem:light-bound}, we get that the better of the two unit
allocations $\tilde{x}_\light$ and $\tilde{x}_\heavy$ provides an $O(\beta)$
approximation to the total fractional value of $x$, where we used the fact that
$\log L/\log \beta = \beta$.

\begin{theorem}
\label{thm:unit-capacity}
    For every fractional allocation $x$ for the unit-capacity setting, there
    exists a fractional unit allocation $x'$ such that
    \[
        \fval(x)\le  O(\log L/\log\log L)\fval(x').
    \]
\end{theorem}

    \subsection{Extension to arbitrary capacities}
\label{sec:arbit-cap-ub}

We will now give a reduction from the setting with arbitrary item
multiplicities to the unit-capacity case discussed in the previous section.
Once again we start with an arbitrary fractional allocation $x$ and construct a
unit allocation $x'$ with large fractional value. The main idea behind the
reduction is to first partition the fractional allocation $x$ into ``layers''.
Layer $r$ corresponds to the $r$th copy of each item (if it exists). Each job
is assigned to one layer, with its fractional allocation appropriately scaled
down, in such a manner that the layers together capture much of the fractional
value of $x$. We then apply Theorem~\ref{thm:unit-capacity} to each layer
separately, obtaining unit allocations for each layer. The union of these unit
allocations immediately gives us a unit allocation overall.

\begin{theorem}
    \label{thm:arbit-capacity}
    For every feasible fractional allocation $x$ in the interval
    preferences setting with arbitrary capacities, there exists a
    fractional unit allocation $x'$, such that
    \[
        \fval(x)\le  O(\log L/\log\log L)\fval(x').
    \]
\end{theorem}

\begin{proof}
  We begin by defining layers. Recall that $B_t$ denotes the number of
  copies of item $t$ that are available. We assume without loss of
  generality that $B_t = \lceil \sum_{j: t\in \intj} x_j\rceil$ for
  all $t$ (and otherwise redefine $B_t$ to be the latter
  quantity). Let $\Bmax = \max_t B_t$. Then we define the $r$th layer
  for $r\in [\Bmax]$ as $\layer:=\{t\in\items:B_t\geq r\}$. 

  We will now construct a fractional allocation $\xr$ for each layer
  $\layer$ with the properties that {\em (i)} for each $r$, $\xr$ is
  feasible with respect to the set of items in $\layer$, and {\em
    (ii)} $\sum_r \fval(\xr) \ge \frac 14 \fval(x)$.

    We proceed via induction over $\Bmax$. For the base case, suppose that
    $\Bmax\le 4$. Then we define $\xr[1] = \frac 14 x$ and $\xr=0$ for all
    $r>1$. Observe that both the properties {\em (i)} and {\em (ii)} are
    satisfied by this definition. For the inductive step, we pick a set $S$ of
    jobs as given by the following lemma (proved in \Cref{sec:deferred}). 

    \begin{lemma}
        \label{lem:greedy-layer}
        For any feasible fractional allocation $x$ in the interval
        preferences setting with arbitrary capacities, one can
        efficiently construct a set $S$ of jobs such that the total
        fractional weight of $x_S$ at any item $t$ is at least
        $\min\{1,B_t\}$ and at most $4$.  Formally, for all items $t$,
        $\min\{1,B_t\}\le \sum_{j\in S: t\in \intj} x_j < 4$.
    \end{lemma}

    Having constructed such a set $S$, we set $\xr[1] = \frac 14 x_S$, and
    recursively construct allocations for the remaining layers using $x-x_S$.
    Observe that $\frac 14 x_S$ is feasible for $\layer[1]$ by the definition
    of $S$. Furthermore, by removing jobs in $S$ from $x$, we reduce $\Bmax$ by
    at least $1$, and therefore, we can apply the inductive hypothesis to
    construct allocations for the remaining layers. This provides us with {\em
    (i)} and {\em (ii)} as desired above.

    Finally, for each layer $\layer$, we apply Theorem~\ref{thm:unit-capacity}
    to the allocation $\xr$ to obtain a fractional unit allocation $\xpr$ for
    that layer. Then, the allocation $x' = \sum_r \xpr$ is a fractional unit
    allocation with $\fval(x)\le O(\beta) \fval(x)$.
\end{proof}

Combining \Cref{thm:arbit-capacity} with
Lemmas~\ref{lem:fopt-upperbound} and \ref{lem:FGL} immediately implies
Theorem~\ref{thm:unit-cap-ub}, which we state here for completeness.

\begin{numberedtheorem}{\ref{thm:unit-cap-ub}}
For the interval preferences setting, there exists a static, anonymous bundle
pricing with competitive ratio $O\left(\frac{\log L}{\log\log
    L}\right)$ for social welfare.
\end{numberedtheorem}

% In the fixed-capacity setting, there exists a static, anonymous bundle
% pricing $(\Pi,\prices)$ such that
% \[
%     \opt \le O\left(\frac{\log L}{\log\log L}\right)\sw(\Pi,\prices).
% \]

    \subsection{The large markets setting}
\label{sec:large-cap-ub}

In this section we consider the setting where every item is available in large
supply. Specifically, let $B:= \min_t B_t$. We show that as $B$ increases, the
approximation ratio achieved by bundle pricing gradually decreases. 

\begin{numberedtheorem}{\ref{cor:large-cap-ub}}
For the interval preferences setting, if every item has at least $B>0$
copies available, then a static, anonymous bundle pricing achieves a
competitive ratio of
\[
    O\left(\frac1B\frac{\log L}{\log\log L - \log B}\right)
\] 
when $B<\log L$, and $O(1)$ otherwise.
% In the fixed-capacity setting, with $B = \min_t\capt < \log L$, there exists
% a static, anonymous bundle pricing $(\Pi,\prices)$ such that
% \[
%     \opt \le O\left(\frac1B\frac{\log L}{\log\log L - \log B}\right)
%         \sw(\Pi,\prices).
% \]
% When $B \ge \log L$, there exists a bundle pricing which gives an
% $O(1)$-approximation.
\end{numberedtheorem}

% \begin{theorem}
%     \label{thm:large-capacity}
%     For the fixed capacity setting with $B= \min_t B_t< \log L$, there exists a
%     static, anonymous bundle pricing that achieves an approximation ratio of
%     \[
%         \frac 1B \cdot \frac{\log L}{\log \log L - \log B}
%     \]
%     for social welfare. When $B\ge\log L$, an $O(1)$ approximation ratio is
%     achieved. Given access to an optimal feasible fractional allocation, the
%     pricing can be constructed in polynomial time.
% \end{theorem}

\begin{proof}
  Let $k=\frac 12\min\{B,\log L\}$ and $\alpha=L^{1/k}$. We will
  partition both the item supply and the jobs into multiple instances,
  such that on the one hand, the fractional solution confined to jobs
  within an instance will be feasible for the item supply in that
  instance; on the other hand, within each instance job lengths will
  differ by a factor of at most $\alpha$. Then, applying
  \Cref{thm:arbit-capacity} will give us a fractional unit allocation
  for every instance with a factor of
  $\Omega\big(\frac{\log\log \alpha}{\log \alpha}\big)$ loss in social
  welfare. Applying \Cref{lem:FGL} to the union of these unit
  allocations then implies the theorem.

  Let $x^*$ be the optimal fractional allocation. Divide all jobs into
  $k$ groups $U_1,\cdots,U_k$ according to length: $U_i$ contains all
  jobs of length between $\alpha^{i-1}$ and $\alpha^i$.  Let $x_i$ denote the
  allocation $x^*$ confined to the set $U_i$ scaled by $1/2$, that is,
  $x_i = \half x^*_{U_i}$.  We now specify the supply for instance $i$. Let
  \[ B_{it} = \left\lceil\half\sum_{j\in U_i:t\in
      I_j}x^*_{j}\right\rceil\]
Note that no item is over-provisioned (although some supply may be wasted) :
    \begin{equation*}
        \sum_{i=1}^{k} B_{it} \leq \sum_{i=1}^{k}
                \left(1+\half\sum_{j\in U_i:t\in I_j}x^*_{j}\right) 
        \leq k+\frac{1}{2}B_t \leq B_t. 
    \end{equation*}
    Note also that $x_i$ is feasible for the supply $\{B_{it}\}$.

    Applying Theorem~\ref{thm:arbit-capacity} and \Cref{lem:FGL} to
    the instances defined above then implies an approximation factor
    of $O\Big(\frac{\log \alpha}{\log\log \alpha}\Big)$, which is $O\Big(\frac
    1B\frac{\log L}{\log\log L-\log B}\Big)$ for $B<\log L$, and
    $O(1)$ otherwise.
\end{proof}

 % Divide the entire supply into $k$
 %  parts, each part assigned to one group of jobs. Denote by $B_{it}$
 %  the number of copies of item $t$ assigned to group $U_i$; that is,
 %  let
 %  $\displaystyle B_{it} = \left\lceil\half\sum_{j\in U_i:t\in
 %      I_j}q_jx^*_{j}\right\rceil$.  Then this is a feasible division
 %  (some supply may be wasted) since

 %    Consider the following fractional allocation rule: each group of jobs can
 %    only be allocated to the corresponding supply, while the allocation is
 %    changed to $x'=\frac{1}{2}x^*$.  This is indeed a feasible fractional
 %    allocation, while the social welfare achieved is half of $\fopt$.  By
 %    Theorem~\ref{thm:arbit-capacity} there is a fractional unit allocation for
 %    each group of jobs which in total achieves at least an $\Omega\big(\frac{\log\log
 %    a}{\log a}\big)$ fraction of $\fopt$. Thus
 %    \begin{equation*}
 %        \fopt\le O\Big(\frac{\log a}{\log\log a}\Big)\uafopt =
 %        O\Big(\frac{\log L}{B\log\log L-\log B}\Big)\uafopt.
 %    \end{equation*}

\section{Lower Bounds}
\label{sec:lowerbound}

We now turn to lower bounds. In this section, we ignore incentive
constraints and bound the social welfare that can be achieved by {\em
  any} online algorithm for the online resource allocation problem.
We consider the case of interval preferences in
Section~\ref{sec:lb-int} and path preferences in
Section~\ref{sec:pathlowerbound}.

\subsection{Lower bound for interval preferences}
\label{sec:lb-int}

We focus on the special case where every buyer is single minded. That
is, for all $i$ and $\vali\sim\disti$, there exists an interval $I$
and a scalar $v$ such that $\vali(I')=v$ for all $I'\supseteq I$ and
$0$ otherwise. This setting is also called {\em online interval
  scheduling}. A lower bound for the unit-capacity version of this
problem was previously developed by \citet{im2011secretary}. Although
\citeauthor{im2011secretary}'s focus was on the secretary-problem
version (i.e. with random arrival order and adversarial values rather
than worst-case arrival order and stochastic values), their lower
bound construction uses jobs drawn independently from a fixed
distribution, and it thus provides a lower bound for the Bayesian
adversarial order setting as well. We restate
\citeauthor{im2011secretary}'s result; our proof of
Theorem~\ref{thm:lb} builds upon this result.

    \begin{lemma}\label{ISSPlemma}
      (Restated from Theorem 1.3 in \cite{im2011secretary}) For the
      unit-capacity setting of the online interval scheduling problem,
      every randomized online algorithm has a competitive ratio of
      $\Omega\left(\frac{\log L}{\log\log L}\right)$.
    \end{lemma}

\begin{numberedtheorem}{\ref{thm:lb}}
  For the online interval scheduling problem, every randomized online
  algorithm has approximation ratio
  $\Omega\left(\frac{\log L}{B\log\log L}\right)$, where
  $B=\min_{t}B_t$.
\end{numberedtheorem}

    % For the case $B=1$, the algorithmic problem we study is exacty the
    % prophet-inequality version of the interval-scheduling problem studied by
    % \citet{im2011secretary}.  Although they studied the secretary-problem
    % version (i.e. random arrival order rather than worst-case), their lower
    % bound construction uses jobs drawn independently from a fixed distribution,
    % and it thus provides a lower bound for the prophet-inequality setting as
    % well.

%They proposed the following construction. Assume that there are $h$ level of nodes
%in an $h^2$-ary tree, each node in level $\ell$ represents 
%an interval of size $h^{2(h-\ell)}$: the set of intervals in level $\ell$ is
%$\{\left((k-1)h^{2(h-\ell)},kh^{2(h-\ell)}\right]:k\in[h^{2(\ell-1)}]\}$. At each
%time, a single-minded job arrives aiming at one of the $H=\sum_{\ell\in[h]}h^{2(\ell-1)}$
%intervals in the tree, with value the same as the length of the interval. The authors
%proved that when $\frac{H}{h}$ jobs come in total, optimal offline allocation 
%can get welfare $\Theta(h^{2(h-1)})$, while any randomized online algorithm can only
%achieve welfare $O(h^{2h-3})$. Thus the approximation ratio is 
%$\Omega(h)=\Omega\left(\frac{\log L}{\log\log L}\right)$, where $L=h^{2(h-1)}$ is the 
%length of longest interval in the tree.

\begin{proof}
  We argue that if there exists an online algorithm with competitive 
  ratio $\alpha=o\left(\frac{\log L}{B\log\log L}\right)$ for the setting with minimum
  capacity $B$, then we can construct an online algorithm for the
  unit-capacity setting with competitive ratio $O(B\alpha) = o\left(\frac{\log
    L}{\log\log L}\right)$, contradicting Lemma~\ref{ISSPlemma} above.

  Our main tool is the following lemma which shows that, for
    any set of intervals, if the intervals arrive in arbitrary order, there
    exists an online algorithm that colors them (i.e. such that no two
    overlapping intervals have the same color) using not many more colors than
    optimal.

    \begin{lemma}\label{coloringlemma}
        (Restated from Theorem 5 in \cite{kierstead1981extremal}) For any set
        of intervals $U$ such that every point belongs to at most $B$ intervals
        in $U$, there exists an online coloring algorithm $\coloralg$ utilizing
        at most $3B-2$ colors. 
    \end{lemma}

    Suppose that for some $B>1$ there exists an online algorithm $\alg$ for
    markets of size $B$ with competitive ratio $\alpha=o\big(\frac{\log L}{B\log\log
    L}\big)$. Consider the following algorithm $\alg'$ for the unit-capacity
    setting. Let $c$ be an integer chosen u.a.r. from $[3B-2]$. As each buyer $i$
    arrives, if $\alg$ accepts $i$, then input $I_i$ to $\coloralg$; if
    $\coloralg$ assigns color $c$ to $I_i$, then accept $i$.

    % \begin{algorithm}[t]
    %     \SetAlgoNoLine
    %     \KwIn{Set of jobs $U$}
    %     Let $c$ be an random integer chosen uniformly at random from 1 to $3B-2$\;
    %     Run $\alg$ on $U$ where the number of copies of each item is set to be
    %     $B_t=B$\;
    %     \For{each job $j$ that arrives
    %     }{
    %         \If{$j$ is accepted by $\alg$ and allocated interval $I_j$}{
    %             input $I_j$ to online interval coloring algorithm $\coloralg$\;
    %             \If{$j$ is assigned to color $c$}{
    %                 $\alg'$ accepts $j$ and allocate $I_j$ to it
    %             }
    %         }
    %         Otherwise reject $j$\;
    %     }
    %     \caption{$\alg'$ for $B_t=1$}
    %     \label{euclid}
    % \end{algorithm}

    % Now we prove that $\alg'$ has approximation ratio $o\big(\frac{\log
    % L}{B\log\log L}\big)$.
    % For any instantiation of jobs, by Lemma
    % \ref{coloringlemma} it is successfully colored with one of the $3B-2$
    % colors.
    By Lemma~\ref{coloringlemma}, the expected value of intervals with color
    $c$ is a $\frac{1}{3B-2}$-fraction of the social welfare obtained by
    $\alg$. Thus the approximation ratio of $\alg'$ for the unit-capacity
    setting is $(3B-2)\alpha = o\big(\frac{\log L}{\log\log L}\big)$, which contradicts
    Lemma~\ref{ISSPlemma}.
    % the inapproximability result of $B=1$.
\end{proof}

\subsection{Lower bound for path preferences}
\label{sec:pathlowerbound}
In this section, we show that for the ``online path scheduling''
problem on trees no online algorithm can achieve subpolynomial
approximation with respect to $L$, the length of the longest possible
path in the instance. 

\begin{theorem}
    \label{thm:treelowerbound}
    There exists an instance of the path preferences setting, where
    every buyer desires a path of length between $1$ and $L$, for
    which every randomized online algorithm has competitive ratio
    $\Omega\left(\sqrt{\frac{L}{\log L}}\right)$ for social welfare.
    % For the path interval scheduling problem on trees, every randomized
    % online algorithm has approximation ratio
    % $\Omega\left(\sqrt{\frac{L}{\log L}}\right)$.
\end{theorem}

\begin{proof}
  Consider a complete binary tree with height $L$ and capacity $1$ on
  every edge in the tree. We will define an instance with a unique
  buyer for every (node, leaf) pair in the tree where the leaf belongs
  to the subtree rooted at the node. Formally, define the level of
  edges bottom-up: for example, the leaf edges (i.e. edges adjacent to
  a leaf node) have level 0, and the edges adjacent to root have level
  $L-1$. There are $n_\ell=2^{L-\ell}$ edges at level $\ell$, and each
  of these edges has $2^{\ell}$ leaves in its subtree. For every level
  $\ell$, every edge $e$ at level $\ell$, and every leaf edge $e'$ in
  the subtree rooted at $e$, our instance contains a unique buyer that
  arrives independently with probability $\frac{1}{2^\ell}$ and has
  value $v_\ell=2^\ell$. We say that the buyer is at level
  $\ell$. Note that there are exactly $2^L$ distinct buyers at each
  level $\ell$, and therefore $L2^L$ buyers in all.

% Say a
%   path ``begins with'' an edge $e$, if $e$ is the only peak edge of
%   the path. For each inner edge $e$ with level $\ell$ in the tree,
%   there are $2^\ell$ distinct paths that goes down to a leaf node
%   beginning with $e$. For every such path, there is a buyer with value
%   $v_\ell=2^\ell$ that demands only this path and arrives
%   independently with probability $\frac{1}{2^\ell}$. Say a buyer is at
%   level $\ell$, if it begins with an edge at level $\ell$.

  Let us now compute the offline optimal welfare $\opt$ for this
  instance. We claim $\opt=\Omega(L2^L)$. Consider a greedy allocation
  that admits buyers in order from the highest to the lowest level:
  every buyer whose path is still available when considered gets
  allocated with probability $1/2$. Note that when a level $\ell$ path
  beginning with edge $e$ is considered, it can be allocated if and
  only if $e$ is not allocated to previous buyers. We now make two
  observations. First, the expected number of buyers at level
  $\ell'>\ell$ whose paths contain $e$ and that arrive is at most
  $\sum_{\ell'>\ell} 2^{\ell}/2^{\ell'} < 1$. So, the probability that
  an edge $e$ is allocated prior to considering buyers at its level is
  at most $1/2$. Second, if prior to considering level $\ell$ buyers,
  edge $e$ is available, then the probability that it gets allocated
  to a level $\ell$ buyer is at least
  $1-(1-\frac{1}{2^{\ell+1}})^{2^\ell}\geq 1-\frac{1}{\sqrt{\textrm{e}}}$. Therefore, the total
  contribution of level $\ell$ buyers to the social welfare is at
  least $\frac{1}{2}\left(1-\frac{1}{\sqrt{\textrm{e}}}\right)n_\ell v_\ell=\Omega(2^L)$. This proves our claim.

% Before buyers at level $\ell$ arrive, at most half
%   of the edges are allocated. For each unallocated edge $e$ at level
%   $\ell$, since there are $2^\ell$ buyers beginning with it arriving
%   and each with probability $\frac{1}{2^\ell}$, then with probability
%   $1-(1-\frac{1}{2^\ell})^{2^\ell}\geq 1-\frac{1}{e}$ one buyer
%   beginning with $e$ will get allocated. Thus the total expected
%   welfare contributed from buyers at level $\ell$ is at least
%   $\frac{1}{2}\left(1-\frac{1}{e}\right)n_\ell
%   v_\ell=\Omega(2^L)$. Therefore the total welfare achieved by this
%   greedy algorithm is $\Omega(L2^L)$, thus $\opt=\Omega(L2^L)$.

  Next we prove via a charging argument that the welfare of any online
  algorithm is bounded by $O(2^L\sqrt{L\log L})$. Assume that 
  buyers arrive in order from lowest level to highest level.
  % , and consider any online
  % algorithm. Say a job ``covers'' a leaf edge $e'$, if its beginning
  % edge $e$ is on the path from $e'$ to the root. Note that if a job
  % covering a leaf edge $e'$ is admitted, then no job at a higher level
  % whose path contains $e$ can be admitted. Each job at level $\ell$
  % covers exactly $2^{\ell}=v_\ell$ leaf edges. When a buyer beginning
  % with edge $e$ at level $\ell$ arrives, it can be accepted with
  % probability at most $\frac{m(e)}{2^\ell}$, where $m(e)$ is the
  % number of leaf edges below $e$ left uncovered by previous allocated
  % jobs. 
  Construct a grid with height $L$ and width $2^L$ as follows. Each
  edge at level $\ell$ corresponds to a size $1\times 2^\ell$
  rectangle in the grid at the same level; see Figure
  \ref{fig:lb-example}. With each leaf edge, we also associate the
  entire column above its cell in the grid with the edge. So each cell
  in the grid is indexed by a level and a column corresponding to a
  leaf edge. As buyers arrive and the online algorithm makes
  allocation decisions, we mark cells in the grid to indicate
  availability. Initially all cells are unmarked. Whenever the online
  algorithm allocates a path to a buyer at level $\ell$ with first
  edge $e$, we mark all cells in the $1\times 2^\ell$ rectangle
  corresponding to edge $e$, as well as all cells above this
  rectangle. Now suppose that a buyer at level $\ell$ arrives whose
  path begins with edge $e$ and ends at leaf edge $e'$. If the cell
  corresponding to column $e'$ and level $\ell$ in the grid is already
  marked, this means the buyer cannot be legally allocated because
  a buyer starting at some other edge on the path from $e$ to $e'$ was
  previously allocated. Let $m(e)$ be the number of unmarked cells 
  in the rectangle corresponding to $e$ when buyers beginning with $e$
  start to arrive.
  If the online algorithm allocates to a buyer
  with starting edge $e$, the total number of cells that get marked at
  this iteration is exactly $(L-\ell)m(e)$.

    \begin{figure}[htbp]
        \centering
        \subfigure{
        \begin{minipage}[b]{0.8\textwidth}
            \includegraphics[width=0.48\textwidth]{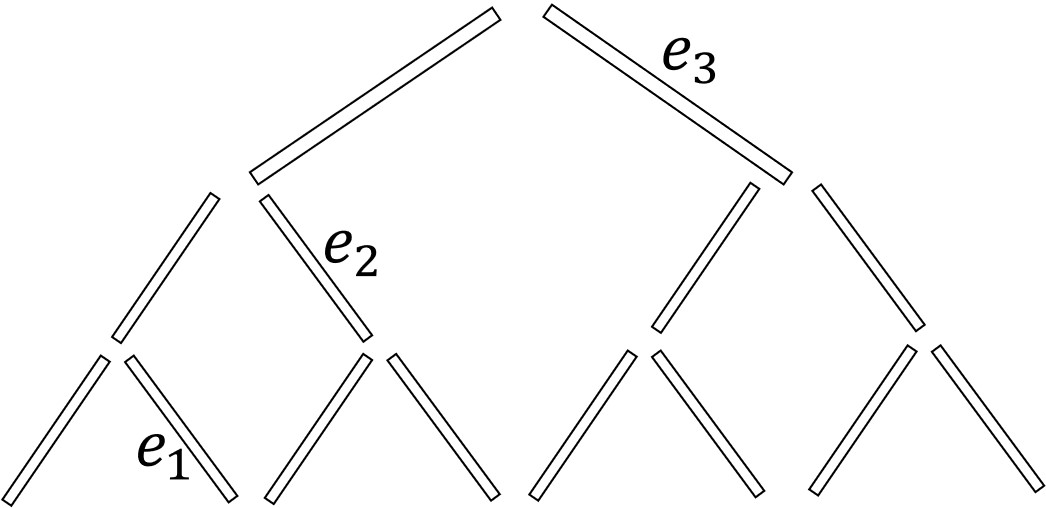}
            \includegraphics[width=0.48\textwidth]{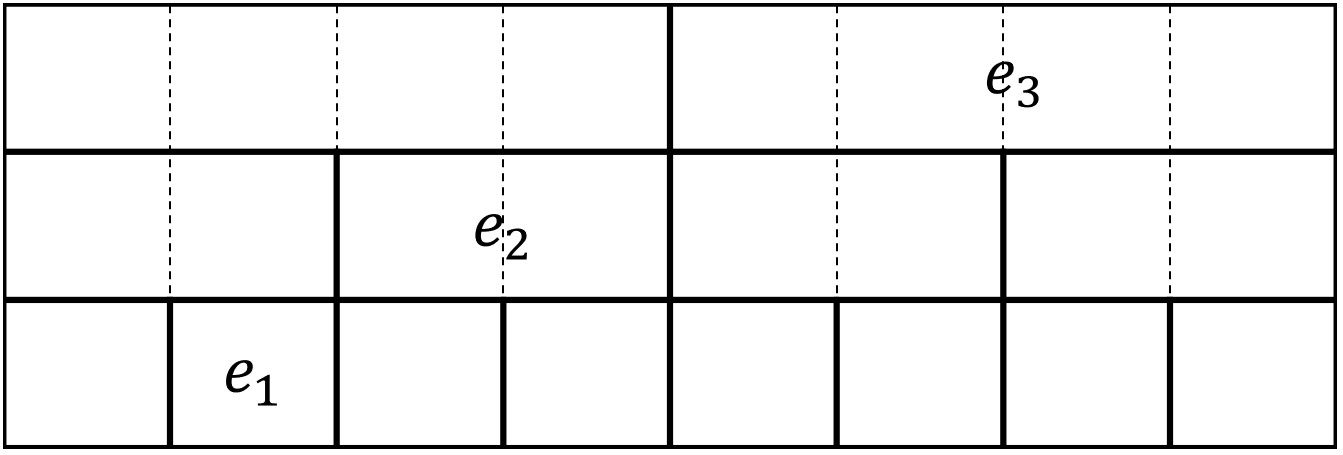}
            \center{(a)}
 %           \caption{\label{fig:lb-example}\small{An example of a binary tree and its
 %             corresponding grid.}}
        \end{minipage}
        }\\
        \subfigure{
        \begin{minipage}[b]{0.8\textwidth}
            \includegraphics[width=0.48\textwidth]{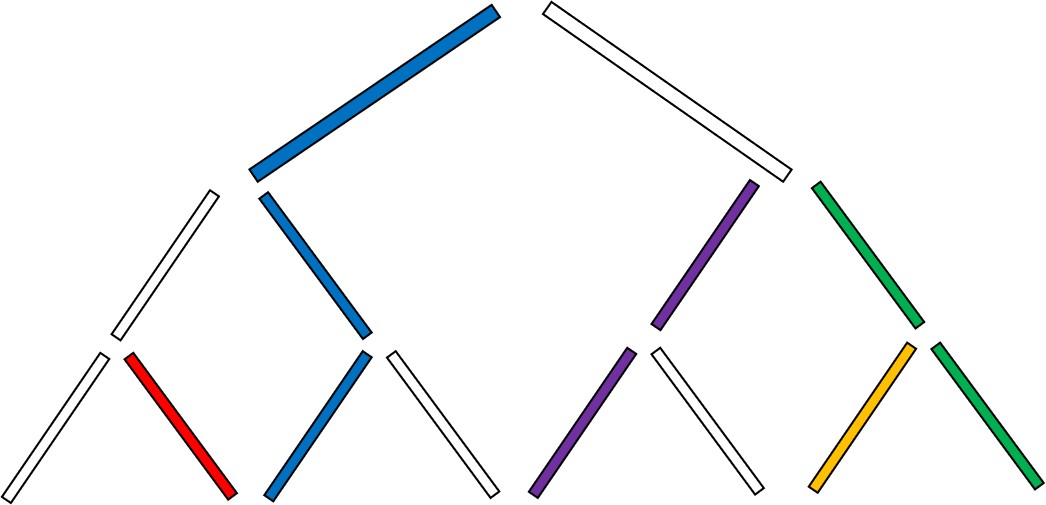}
            \includegraphics[width=0.48\textwidth]{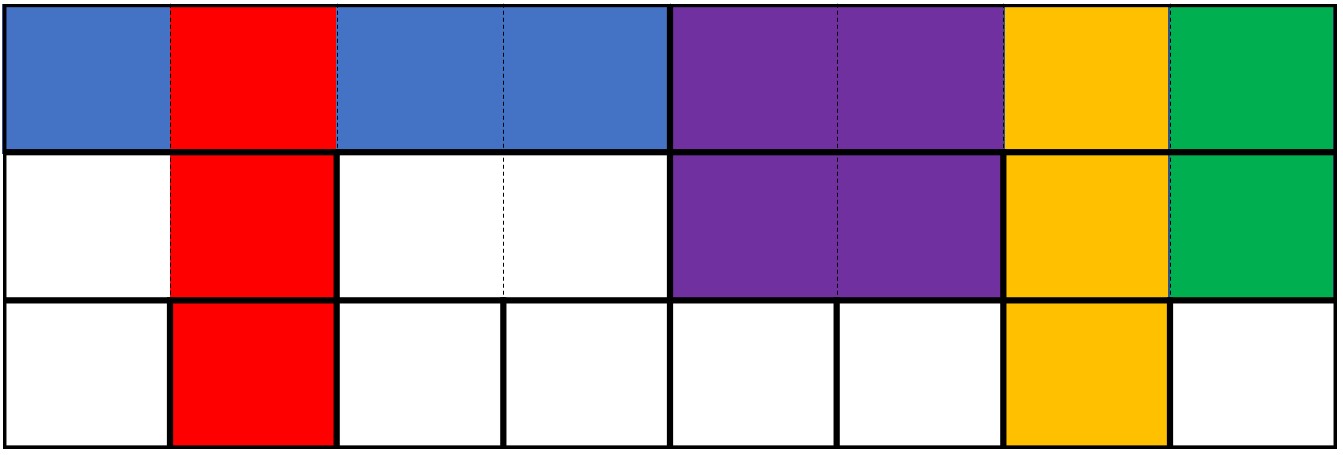}
            \center{(b)}
%            \caption{\label{fig:lb-example-2}\small{Jobs arrive bottom-up and get allocated. The
%              grid shows the cells marked by each job.}}
        \end{minipage}
        }
        \caption{\label{fig:lb-example} \small{(a) An example of a binary tree and its
              corresponding grid; (b) Jobs arrive bottom-up and get allocated. The
              grid shows the cells marked by each job.}}
    \end{figure}

    Say an edge $e$ at level $\ell$ is ``good'' if when the buyers beginning with $e$
    start to arrive, we have $(L-\ell)m(e)\geq 2^\ell\sqrt{L/\log
      L}$. For any buyer beginning with a good edge that is allocated
    by the algorithm, it marks a rectangle in the grid with area
    $(L-\ell)m(e)\geq \sqrt{L/\log L}\cdot v_\ell$, i.e
    $\sqrt{L/\log L}$ times the value of the buyer. Since the total
    area of the grid is $L2^L$, the total value of buyers with
    good edges allocated in the algorithm is at most
    $\frac{1}{\sqrt{L/\log L}}L2^L=2^L\sqrt{L\log L}$.

    Now consider the welfare contributed by buyers at bad edges. At
    any such bad edge $e$ at level $\ell$, at most $m(e)$ buyers
    beginning with $e$ can get allocated; the expected number of such
    buyers that arrive is at most
    $\frac{m(e)}{2^\ell}<\frac{\sqrt{L/\log L}}{L-\ell}$. Thus their
    total welfare contribution is at most
    \begin{equation*}
      \sum_{\ell=0}^{L-1}n_\ell v_\ell\cdot\frac{\sqrt{L/\log L}}{L-\ell}=2^L\sqrt{L/\log L}\sum_{\ell=0}^{L-1}\frac{1}{L-\ell}=O(2^L\sqrt{L\log L}).
    \end{equation*}

    Summing up the total welfare contribution from buyers at good and
    bad edges gives us a bound of $O(2^L\sqrt{L\log L})$ on the total
    welfare obtained by any online algorithm.
\end{proof}

\section{Bundle Pricings for Path Preferences}
\label{sec:trees-ub}

We now turn to the path preferences setting. Recall that here the
items correspond to edges in a fixed tree. Each buyer $i$ desires a
path $P_i\subset T$. Our techniques for this setting are similar to
those in Section~\ref{sec:ub-interval} --- we construct a partition of
the items into layers, construct a corresponding fractional allocation
that respects this layer structure in that every buyer is allocated
items from at most two layers, and no layer has too much fractional
weight. We show in Section~\ref{sec:layered-alloc} that such a good
{\em fractional layered allocation} always exists. In
Section~\ref{sec:trees-unit-cap} we show that a layered allocation
guarantees a good bundle pricing. In
Section~\ref{sec:trees-large-cap-ub} we present the corresponding
``large markets'' result, namely that the competitive ratio decreases
linearly with capacity.

\subsection{Fractional layered allocations}
\label{sec:ub-tree-layered}

% Here, however, we must deal with branching in
% the graph, which introduces additional challenges in finding a layered
% fractional solution and complicates charging the welfare of blocked jobs.
% \bmnote{this doesn't make sense if we're using DFKL for intervals: rewrite}

As in the case of fractional unit allocations that we defined for the interval
preferences setting, we would like to construct a partition of items into
bundles and a corresponding allocation of jobs such that each job fits cleanly
into one bundle, no bundle contains too much fractional weight, and no item
belongs to more bundles than its multiplicity. Unfortunately, it turns out that
due to the rich combinatorial structure among paths in trees, such an
allocation cannot be constructed without losing too much value in the worst
case.
% Consider, in particular, the following example.
% \begin{example}
%     \label{ex:star-layering}
%     Let $T$ be a star graph with $n+1$ edges. Edge $e_0$ has capacity $n$, and
%     all other edges have capacity 1. There are $n$ jobs with weight 0.9 each,
%     and job $j \in [n]$ desires the path $(e_0, e_j)$. There are $n-1$
%     additional jobs with weight 0.5 each such that, for each $j \in [n-1]$, one
%     of these jobs desires the path $(e_j,e_{j+1})$. There is no way to
%     decompose this into bundles such that each job belongs to one bundle, each
%     bundle has constant fractional weight, and each item is associated with at
%     most one bundle.
% \end{example}

To overcome this issue, we simplify the combinatorial structure by
orienting the underlying tree with an arbitrary root and breaking each
path into two ``monotone'' components. Formally, for an arbitrary
rooting of the underlying tree, define the depth of an edge $t\in T$,
denoted $d(t)$, to be the number of edges on the shortest path from
the root to this edge, including $t$ itself.\footnote{So, the edges
  incident on the root have depth $1$.} For a job $j$ with
corresponding path $P_j$, define the {\em peak} of $j$'s desired path
to be the edge(s) closest to the root:
$\peakj = \argmin_{t \in \intj} \{d(t)\}$. Let $\larmj$ and $\rarmj$
denote the two subpaths of $\intj$ that begin at one of the peak edges
in $\peakj$ and consist of all subsequent edges in $\intj$ with
increasing depth; we call these the {\em arms} of the job. Let
$P_{\larmj}$ and $P_{\larmj}$ denote the corresponding subpaths. Note
that one of these arms may be empty; in that case, we take $\rarmj$ to
be the empty arm.

\begin{definition} A fractional allocation $x$ is a {\em fractional
    layered allocation} if there exists a partition of the multiset of
  items (where item $t$ has multiplicity $\capt$) into bundles
  $\{T_1, T_2, T_3, \cdots\}$, and a corresponding partition of the
  arms with non-zero weight under $x$, $\cup_{j\in U: x_j>0} \{\larmj,\rarmj\}$,
  into sets $\{A_1, A_2, A_3, \cdots\}$, such that:
    \begin{itemize}
        \item For all $j\in U$ with $x_j>0$, there is exactly one
          index $k$ with $\larmj\subset A_k$ and, if $\rarmj$ is not
          empty, exactly one index $k'$ with $\rarmj\in A_{k'}$. 
          \item For all $k$ and $\armj\in A_k$, $P_{\armj}\subseteq T_k$. 
        \item For all $k$ and $t\in T_k$, we have $\fwt(x,
          \{j: \armj\in A_k \text{ and } P_{\armj}\ni t\}) \le 1$.
   \end{itemize}
\end{definition}

Observe that there are two main differences between the layered
allocations defined above and fractional unit allocations as defined
in Section~\ref{sec:frac-unit-alloc}. First, layered allocations partition arms into layers
rather than entire jobs, so a job can end up in two different
layers. Second, and more importantly, a layer can have very large
total weight. All we guarantee is that the weight of any single edge
in a layer will be bounded by $1$. As such, the pricing we construct
is allowed to allocate multiple subpaths of a single layer to
different buyers in
sequence.

% We can now define a layered fractional allocation for the tree setting.
% \begin{definition}
%     Given a tree $T$ and jobs $U$, a fractional allocation $y$ is a {\em
%     layered fractional allocation} if there exist layers
%     $\{\layeri[1],\layeri[2],\cdots\}$ which satisfy the following conditions:
%     \begin{itemize}
%         \item for all $j$ with $y_{\larmj} > 0$, $\larmj \in \layeri$ (and
%             $\rarmj \in \layeri[i']$) for exactly one $i$ (and exactly one
%             $i'$, if $\rarmj$ is not empty)
%         \item for all $i$ and all $t$, $w_t(\layeri) \le 1$
%         \item for all $t$, $|\{i : w_t(\layeri) > 0\}| \le B_t$
%     \end{itemize}
% \end{definition}

% Given a fractional allocation $x$ over jobs, we extend our definitions of
% $\fwt$ and $\fval$ to arms. We use the notation $y$ to indicate a fractional
% allocation over arms. Given a feasible allocation $x$ over jobs, the allocation
% $y_{\larmj} = y_{\rarmj} = x_j$ is a feasible fractional allocation. For a job
% $j$, if both arms are defined, we take the value associated with each arm to be
% half the value of the corresponding job: thus, the total value of the two arms
% is equal to the value of the job. That is, $\valj[\larmj] = \valj[\rarmj] =
% \half \valj$. If $\rarmj$ is empty, we take $\valj[\rarmj] = y_{\rarmj} = 0$.
% For a set of arms $S$ and an edge $t$, define $w_t(S) = \sum_{\armj \in S}
% y_{\armj}$.  Finally, define
% \[
%     \fval(y) = \sum_j \big(\valj[\larmj] y_{\larmj} + \valj[\rarmj]
%     y_{\rarmj}\big).
% \]

\subsection{Layering a tree}
\label{sec:layered-alloc}

We now show that good fractional layered allocations always exist. For
a fractional allocation $x$, let $\vmax(x) = \max_{j: \xj>0} v_j$ and
$\vmin(x) = \min_{j: \xj>0} v_j$. 

\begin{lemma}
    \label{thm:layered-tree}
    For every feasible fractional allocation $x$, one can efficiently
    construct a feasible layered fractional allocation $\tilde{y}$
    such that
    \[
        \fval(x) \le O\left(\frac{\vmax(x)}{\vmin(x)}\right) \fval(\tilde{y}).
    \]
\end{lemma}

Before we proceed we need more notation. For a subset of jobs (or
arms) $J$, let $E_J$ denote the set of edges collectively desired by
those jobs: $E_J = \cup_{j\in J} P_j$. For an edge $t$, let
$\fwe(x) := \fwt(x,\{j:P_j\ni t\})$ denote the total fractional weight
of jobs (or arms) whose paths contain $t$. Likewise, let $\fwe(x_A)$
or $\fwe(x,A)$ denote the fractional weight on $t$ from jobs (or
arms) in set $A$.

We construct the layering recursively. Informally, given a fractional
allocation $x$, we find a set of jobs $\tlayer$ such that the fractional demand
on any item desired by a job in $\tlayer$ is a constant. That is, for
all $t\in E_\tlayer$, $\fwe(x,\tlayer)\le c$. 
% let $T_\tlayer
% = \Union_{j \in \tlayer} \intj$; then, for all $t \in T_\tlayer$, $w_t(\tlayer)
% \le c$. 
We scale the
fractional allocation of jobs in $\tlayer$ by $1/c$, and take this scaled
allocation along with the set $E_\tlayer$ to be the first layer. We thus obtain a feasible
layer, and lose a factor of at most $c$ in fractional value.

However, in order to apply this step recursively and arrive at a
feasible layered allocation, we need the remaining fractional solution
to be feasible for the items remaining after reducing the capacity of
every item in $E_\tlayer$ by one. Therefore, we need the fractional
weight on every edge in $E_\tlayer$ to decrease by at least 1 (or else
to zero) after removing $\tlayer$.  It may not be possible to find a
set $\tlayer$ which satisfies both this and the constant-demand
condition described above. Instead, we find a set $D$ such that jobs
in $\tlayer$ and $D$ together account for one unit of demand on every
edge in $E_\tlayer$, and then we drop $D$.  Because we drop jobs in
$D$ rather than assigning them to a layer, we do not need to ensure
that the fractional weight on edges in $E_D$ decreases by any
particular amount. We pick a set $D$ with small fractional weight
relative to that of $\tlayer$, so that we can charge its lost
fractional value to $\tlayer$.
%, and this step loses at most a factor of
%$O\left(\frac{\vmax(x)}{\vmin(x)}\right)$.

Finally, although we have informally described the argument in terms
of {\em jobs}, in our formal argument we consider sets of {\em
  arms}. This introduces an extra complication: when we drop the set
$D$ of arms, we must also drop their sibling arms. We show that we can
do this without losing more than another constant factor in fractional
value.

We formalize the recursive step in the following lemma; the proof
is deferred to Section~\ref{sec:deferred}.
\begin{lemma}
    \label{lem:peeling}
    Given a set $U$ of jobs and a non-zero fractional allocation $y$,
    there exist sets of arms $\tlayer \ne \emptyset$ and $D$ such that 
    \begin{enumerate}[label=\roman*.]
        \item $\tlayer \intersect D = \emptyset$;
        \item $\tlayer \intersect D' = \emptyset$, where $D' = \{\armj : j^{a'}
            \in D, a\ne a'\}$;
        %\item $y'_{\tlayer \union D} = 0$,
        %\item $y'$ is feasible for the multiset of items $T \setminus T_l$,
        \item for all $t \in E_\tlayer$, $\fwe(y, \tlayer \union D)\ge
          \min\{1, \fwe(y)\}$; \label{item:lowerbound}
          % $w_t(\tlayer \union D, y) \ge \min\{1, w_t(U, y)\}$
        \item for all $t \in E_\tlayer$, $\fwe(y, \tlayer) \le 7$; and
          % $w_t(\tlayer, y) \le 7$; and
            \label{item:upperbound}
        \item $\fwt(y_\tlayer) \ge 2 \fwt(y_D)$. \label{item:wtbound}
    \end{enumerate}
    % where $T_\tlayer = \{t : w_t(\tlayer, y) > 0\}$. 
    Furthermore, $\tlayer$ and $D$ can be found efficiently.
\end{lemma}

\begin{proofof}{\Cref{thm:layered-tree}}
  Let $E$ denote the {\em multiset} of edges, with the given
  multiplicities, in the given tree. Given a set of arms $U$ and a
  feasible fractional allocation $x$, we apply \Cref{lem:peeling} to
  generate the sets of arms $\tlayer$ and $D$.  Let $D'$ be the set of
  ``sibling'' arms of arms in $D$. Recall that $E_\tlayer$ is the set
  of edges contained in the paths desired by $j\in \tlayer$. Let
  $E'=E\setminus E_\tlayer$, and $x' = x-x_\tlayer-x_D-x_{D'}$.

  First, we observe that $x'$ is feasible for the multiset $E'$ of
  edges. This is because on the one hand, the multiplicity of each
  edge $t$ in $E_\tlayer$ decreases by 1. On the other hand,
  $\fwe(x)-\fwe(x')\ge \fwe(x,\tlayer\cup D)\ge \min\{1, \fwe(x)\}$ by
  property {\it \ref{item:lowerbound}} in \Cref{lem:peeling}. So,
  either the fractional weight on $t$ decreases by $1$, or it is equal
  to $0$.

  We set $T_1 = E_\tlayer$, $\tilde{A}_1 = \tlayer$, $D_1=D$, and
  $D'_1=D'$, and construct the remaining partition of edges and arms
  by applying \Cref{lem:peeling} recursively to allocation
  $x'$. Observe that our rescursive applications may discard in the
  form of the sets $D'_k$ some arms that have previously been included
  in the sets $\tilde{A}_{k'}$ for $k'<k$. Define
  $A_k = \tilde{A}_k\setminus (\cup_{k'} D'_{k'})$. We then define
  $\tilde{y}_j = x_j/7$ for jobs $j$ with arms in $\cup_k A_k$, and
  $0$ otherwise. It is now easy to see that $\tilde{y}$ is a
  fractional layered allocation, where the third property follows from
  property {\it \ref{item:upperbound}} in \Cref{lem:peeling}.

  It remains to prove that the fractional value of $\tilde{y}$ is
  large enough. Property {\it \ref{item:wtbound}} in
  \Cref{lem:peeling} tells us that
  $\fwt(x, \tilde{A}_k)\ge \fwt(x, D_k\cup D'_k)$. Summing over $k$
  and removing the contribution of $D'_k$ from each side, we get
  \[\fwt(x,\cup_k A_k) \ge \fwt(x,\cup_k D_k) \ge \half\fwt(x,\cup_k
  (D_k\cup D'_k)),\] 
  which implies $\fwt(x,\cup_k A_k)\ge \frac 13 \fwt(x)$, or
  $\fwt(\tilde{y})\ge \frac 1{21} \fwt(x)$.

  Finally, since the values of all jobs with positive weight under $x$
  lie in the range $[\vmin(x),\vmax(x)]$, we get that
  $\fval(\tilde{y})\ge \Omega\left(\frac{\vmin(x)}{\vmax(x)}\right)\fval(x)$.
\end{proofof}

	\subsection{From layered allocations to bundle pricings}
\label{sec:trees-unit-cap}

In this section we will demonstrate the existence of a good bundle
pricing given any fractional layered allocation:

\begin{lemma}
\label{lem:single-value-class}
    Given a feasible fractional layered allocation $y$, there exists a bundle
    pricing $\prices$ such that
    \[
        \fval(y) \le O\left(\frac{\vmax(y)}{\vmin(y)}\right) \sw(\prices).
    \]
\end{lemma}

To develop intuition for our approach consider the special case where
every edge in the tree has capacity $1$ (that is, the layered
allocation has a single layer), and every job has equal value. We
claim that essentially the most contentious edges in a job's path are
its peak edges -- those closest to the root. In particular, if a job
gets ``blocked'' by another admitted job on one or more of its peak
edges, then we should try to recover its lost value in the form of
revenue from the blocking edge. What if a job gets blocked on a
non-peak edge? We can then charge the lost value of this job to the
value of the blocking job---it is easy to observe that any such
blocking job is charged at most once in this manner. In effect, every
job has up to two important edges, namely its peak edges, that
determine whether and to what extent the job will contribute to the
solution's revenue or its utility. Upon focusing on these two edges,
then, the instance behaves like one where every desired bundle has
size $2$ and we can construct a bundle pricing that achieves a
competitive ratio of $O(1)$. In order to convert this argument into a
proof we need to deal with several complications: jobs have different
values; different arms of a job may be allocated at different layers
in the fractional layered allocation; etc.

\begin{proof}
We start by introducing some notation. Given the fractional layered
allocation $y$, let $\{T_1, T_2, \cdots\}$ be the partition of
items/edges into layers, and $\{A_1, A_2, \cdots\}$ be the
corresponding partition of arms. Henceforth we will think of the
different copies of an edge as distinct items, and when we refer to
some edge/item $e$ it will be understood that this edge corresponds to
some particular set $T_k$. For a $j\in U$ with $y_j>0$, and for the
index $k$ such that $\armj\in A_k$, we will redefine $P_{\armj}$ (and
likewise, $P_j$) as the subset of edges of $T_k$ that correspond to
this arm.\footnote{That is, $P_{\armj}$ corresponds to the specific
  copies of edges in $T_k$ and not any copies that form the path
  corresponding to this arm.} Accordingly, we also redefine the peak
edges, $\peakj$, of a job $j$ with $y_j>0$ as the copies of its peak
edges that belong to $T_k$ corresponding to $\armj\in A_k$. Note that
we may assume without loss of generality (via rescaling, e.g.) that
$\vmin(y)=1$.

We are now ready to define our pricing:
\begin{itemize}
\item For any layer $k$, edge $e\in T_k$, and job $j$ such that
  $\armj\in A_k$, define the contribution of $j$ to $e$'s price as
  $\contrib_{ej} = \frac 14 y_j v_j \ind[e \in \peakj]$. For all other
  pairs $(j, e)$, set $\contrib_{ej} = 0$.
\item For any $k$ and edge $e\in T_k$, define the price of $e$
  as $\phat_e = \sum_j \contrib_{ej}$.
\item For any $k$ and a bundle $B\subset T_k$ that corresponds to a
  ``monotone'' path,\footnote{By a monotone path we mean one that has
    a single peak edge.} define the price of the bundle as
  $p_B = \max(1, \sum_{e\in B} \phat_e)$. All other bundles are priced at
  infinity.
\end{itemize}

\noindent
Our mechanism offers the static, anonymous bundle prices $\prices$ as
defined above. Now consider any instantiation of job arrivals. We
consider three events related to job $j$:
\begin{itemize}
\item $\mathcal E^1_j$ occurs if, at the time of $j$'s arrival, all edges in
  $\peakj$ are unsold;
\item $\mathcal E^2_j$ occurs if, after all jobs have arrived, at least one of
  the edges in $\peakj$ is sold;
\item $G_j$ occurs when either of the above holds: $G_j = {\mathcal
    E}^1_j \union {\mathcal E}^2_j$.
\end{itemize}

We now make the following two claims. The first bounds the
contribution to the welfare of the pricing from ``good'' jobs for
which event $G_j$ occurs. The second bounds the loss from the ``bad''
jobs. Here $\util(\prices)$ denotes the total buyer utility generated
by the pricing $\prices$ and $\rev(\prices)$ denotes the total revenue
generated.

\begin{claim}
\label{claim:good-bound}
\begin{equation}
  \label{eq:good-bound}
     \util(\prices) + 4\;\rev(\prices) + \sw(\prices) \ge \sum_jy_j\valj\prob{G_j} - \half\fval(y). 
\end{equation}
\end{claim}

\begin{claim}
\label{claim:bad-bound}
\begin{equation}
  \label{eq:bad-bound}
  2\vmax(y)\sw(\prices) \ge \sum_j y_j\valj \prob{\overbar{G_j}}. 
\end{equation}
\end{claim}

\noindent
Given the two claims, adding Equations~\eqref{eq:good-bound} and
\eqref{eq:bad-bound}, and using the fact that
$\sw(\prices) = \util(\prices) + \rev(\prices) $, we have
    \begin{align*}
        (2\vmax(y) + 5)\;\sw(\prices) &\ge \sum_j y_j\valj - \half\fval(y)
        = \half\fval(y)
    \end{align*}
The theorem follows. It remains to prove the claims.

\begin{proofof}{Claim~\ref{claim:good-bound}}
  For a job $j$ with $y_j>0$, let $p(j) = \sum_{e\in P_j} \phat_e$
  denote the price of the ``intended'' bundle for $j$. We first write
  the utility of the pricing $\prices$:
  \begin{align}
    \util(\prices) &\ge \sum_j y_j \prob{\mathcal E^1_j}
                     \big(\valj - \max\{1, p(j)\}\big) \nonumber \\
                   &= \sum_j y_j \prob{\mathcal E^1_j} \big(\valj - p(j)\big) -
                     \sum_j y_j \prob{\mathcal E^1_j}\big(1 - p(j)\big)^+\nonumber\\
                  &= \sum_j \valj y_j \prob{\mathcal E^1_j} - \sum_j
                    \big(y_j \prob{\mathcal E^1_j} \sum_{e\in P_j} \phat_e\big) -
                     \sum_j y_j \prob{\mathcal E^1_j}\big(1 - p(j)\big)^+
                     \label{eq:pricing-util-lb}
  \end{align}
  where the inequality follows because each job is a profit-maximizer
  and, when not blocked, will pay the sum of the prices (or 1) if this
  is less than its value.

  We bound the three terms in \eqref{eq:pricing-util-lb} separately,
  beginning with the third. Note that job $j$ always purchases some
  bundle containing its path when it arrives, if event $\mathcal E^1_j$ occurs. Therefore,
  since $(1 - p(j))^+ \le 1$ and the single-minded buyer's value for
  any bundle containing his path is least 1,
    \begin{equation} 
       \label{eq:price-floor-bound}
        \sum_j y_j\prob{A_j}\big(1 - p(j)\big)^+ \le \sw(\prices).
    \end{equation}

\noindent
    For the second term, we have:
    \begin{align}
      \sum_j \left(y_j \prob{\mathcal E^1_j} \sum_{e\in P_j} \phat_e\right) 
       \le \sum_e \left(\phat_e \sum_{j:e\in P_j} y_j \right) & \le \sum_e \phat_e \nonumber\\
      & = \sum_{e,j} \contrib_{e,j} = \sum_j \sum_{e\in\peakj} \frac
        14 y_jv_j 
      \le \frac 12 \fval(y)
       \label{eq:sec-term-bound}
    \end{align}
    \noindent
    where the second inequality follows from the definition of
    fractional layered allocations.

    We now bound the first term in \eqref{eq:pricing-util-lb}. For
    edge $e$, let $\mathcal E^2_e$ denote that, after all jobs have
    arrived, edge $e$ is sold. Recall that
    $G_j = \mathcal E^1_j \cup \mathcal E^2_j = \mathcal E^1_j \cup
    (\cup_{e\in\peakj} \mathcal E^2_e) $. Therefore,
    \begin{align}
      v_jy_j\ind[G_j] & \le v_jy_j\ind[\mathcal E^1_j] +
      \sum_{e\in\peakj} v_jy_j\ind[\mathcal E^2_e]\nonumber\\
      & = v_jy_j\ind[\mathcal E^1_j] +
      \sum_{e} 4\contrib_{ej}\ind[\mathcal E^2_e]\nonumber\\
\intertext{Summing over all $j$,}
      \sum_j v_jy_j\ind[G_j] & \le \sum_j  v_jy_j\ind[\mathcal E^1_j]
        + 4\sum_e \ind[\mathcal E^2_e]\sum_j \contrib_{ej}\nonumber\\
      & = \sum_j  v_jy_j\ind[\mathcal E^1_j] + 4 \sum_e 
      \ind[\mathcal E^2_e]\phat_e\nonumber\\
      & \le \sum_j  v_jy_j\ind[\mathcal E^1_j] + 4 \rev(\prices)
\label{eq:first-term-bound}
    \end{align}
    The claim now follows by putting together
    Equations~\eqref{eq:pricing-util-lb},
    \eqref{eq:price-floor-bound}, \eqref{eq:sec-term-bound}, and
    \eqref{eq:first-term-bound}.
\end{proofof}

\begin{proofof}{Claim~\ref{claim:bad-bound}}
  Fix a particular instantiation of jobs and consider a job $j$ for
  which $G_j$ does not occur. Then $j$ is blocked at the time of its
  arrival on some non-peak edge in $P_j$. Let $e$ be the edge closest
  to the root among all such blocking edges. Then, $e$ is a peak edge
  for the job $j'$ that is allocated this edge. We charge the quantity
  $v_jq_j$ to $j'$.

  Now observe that all of the jobs $j$ that are charged to some job
  $j'$ along edge $e\in T_k$ contain the parent edge of $e$, call it $e'$, in their
  arm $\armj\in A_k$. Then we can use the fact that
  $\fwt(y, \{j: \armj\in A_k \text{ and } P_{\armj}\ni e'\}) \le 1$ from the
  definition of fractional layered allocations to assert that the
  total weight of such jobs is at most $1$. In other words, the total
  charge on $j'$ is at most 
  \[\sum_{j: j\text{ charges }j'} v_j q_j \le \vmax(y) \sum_{j: j\text{ charges }j'} q_j \le 2\vmax(y) \le 2\vmax(y) v_{j'}\]
  where the factor of $2$ arises from the fact that $j'$ could be
  charged along both of its peak edges, and the last inequality follows
  by recalling that $v_{j'}\ge \vmin(y)=1$ for any job $j'$ allocated
  by the pricing.

  The claim now follows by summing the above inequality over jobs $j'$
  allocated by the pricing, and taking expectations of both sides over
  possible instantiations.
\end{proofof}
This completes the proof of the lemma.
  \end{proof}

Armed with this lemma, we can now complete the proof of
Theorem~\ref{thm:trees-ub}. Start with the optimal fractional
allocation $y$, and partition all jobs into $\log H$ value classes,
where each class contains jobs with values within a factor of $2$ of
each other. One of these classes, call it $C$, contributes more than a
$\log H$ fraction to the fractional value of $y$. Then, applying
Lemma~\ref{lem:single-value-class} to the allocation $y_C$ gives
us the following theorem.

\begin{numberedtheorem}{\ref{thm:trees-ub}}
For the path preferences setting, there exists a static, anonymous bundle
pricing with competitive ratio $O\left(\log H\right)$ for social welfare.
\end{numberedtheorem}

	\subsection{The large capacity setting}
\label{sec:trees-large-cap-ub}

In this section we consider the setting where every edge is available in large
supply. Recall that we define $B:= \min_t B_t$. 
We show that as $B$ increases, the
approximation ratio achieved by bundle pricing gradually decreases. 

\begin{numberedtheorem}{\ref{cor:tree-large-cap-ub}}
For the path preferences setting, if every edge has at least $B>0$
copies available, then a static, anonymous bundle pricing achieves a
competitive ratio of
\(
    O\left(\frac1B\log H\right)
\)
for social welfare.
% In the fixed-capacity setting, with $B = \min_t\capt < \log L$, there exists
% a static, anonymous bundle pricing $(\Pi,\prices)$ such that
% \[
%     \opt \le O\left(\frac1B\frac{\log L}{\log\log L - \log B}\right)
%         \sw(\Pi,\prices).
% \]
% When $B \ge \log L$, there exists a bundle pricing which gives an
% $O(1)$-approximation.
\end{numberedtheorem}

% \begin{theorem}
%     \label{thm:large-capacity}
%     For the fixed capacity setting with $B= \min_t B_t< \log L$, there exists a
%     static, anonymous bundle pricing that achieves an approximation ratio of
%     \[
%         \frac 1B \cdot \frac{\log L}{\log \log L - \log B}
%     \]
%     for social welfare. When $B\ge\log L$, an $O(1)$ approximation ratio is
%     achieved. Given access to an optimal feasible fractional allocation, the
%     pricing can be constructed in polynomial time.
% \end{theorem}

\begin{proof}
  Let $k=\frac 12 B$ and $\alpha=H^{1/k}$.  The proof technique is
  basically identical to that of \Cref{cor:large-cap-ub}, where we
  partition both the item supply and the jobs into $k$ instances, such
  that on the one hand, the fractional solution confined to jobs
  within an instance will be feasible for the item supply in that
  instance; on the other hand, within each instance job values will
  differ by a factor of at most $\alpha$.  Then, applying
  \Cref{thm:trees-ub} will give us a static, anonymous pricing for
  every instance individually with a factor of
  $\Omega\big(\frac{1}{\log \alpha}\big)$ loss in social welfare. Now
  consider running these instances in parallel, with each buyer
  allowed to purchase bundles from any of the instances. The
  utility-revenue analysis in \Cref{lem:single-value-class} continues
  to apply, providing the same guarantees on welfare. However, we may
  double count some part of the welfare: suppose a buyer $j$ is
  ``assigned'' by the layered allocation to one instance, but
  purchases a bundle in another instance; then we may account for both
  his contribution to utility for the first instance as well as his
  contribution to revenue for the second instance. This double
  counting only costs us a factor of $2$ in the overall welfare, and
  we obtain an $O(\log\alpha)=O\left(\frac{1}{B}\log H\right)$
  approximation.
\end{proof}
\section{Convex Production Costs with Interval Preferences}
\label{sec:costs}

In this section we show that our results for interval preferences, in
particular Theorems~\ref{thm:unit-cap-ub} and \ref{cor:large-cap-ub},
generalize to the setting where the seller can obtain additional
copies of each item at increasing marginal costs. In particular, we
assume the seller incurs cost $\cti$ when allocating\footnote{Note
  that the seller can acquire items online; supply need not be
  provisioned before demand is realized.} the $i$th copy of item $t$,
and $\cti \ge \cti[ti']$ for $i > i'$.  Note that this generalizes the
fixed-capacity setting by taking $\cti = 0$ for $i \le B_t$ and
$\infty$ otherwise.  We begin by revisiting our definition of
fractional allocations and associated notation, which we defined for
fixed-capacity settings in \Cref{sec:prelim}.

% Our most general results apply to a setting in which the seller can acquire
% additional copies of each item at non-decreasing cost. 

In this setting, a fractional allocation $x$ no longer need satisfy an explicit
supply constraint, but the fractional value achieved by $x$ is different, since
the costs of supply must be taken into account.  We begin by introducing
notation for the cost of a given fractional allocation. For any fractional
allocation $x$, denote by $b_t(x)=\sum_{j:\intj \ni t}x_j$ the (fractional)
number of copies of item $t$ allocated in $x$. Let $b_{tr}(x)$ be the fraction
of copy $r$ of item $t$ sold:
\[
    b_{tr}(x) = \begin{cases}
        1 & r\leq b_t(x) \\
        b_t(x)-\lfloor b_t(x)\rfloor & r = \lfloor b_t(x)\rfloor+1 \\
        0 & r>\lfloor b_t(x)\rfloor+1.
    \end{cases}
\]
For simplicity we use $b_t$ to denote $b_t(x)$ and $b_{tr}$ to denote
$b_{tr}(x)$ in later analysis. Observe that the total cost for all copies of
item $t$ in $x$ is $\sum_{r}b_{tr}c_{tr}$.  Then the fractional value achieved
by $x$ is
\[
    \fval(x,\costs) = \sum_j v_jx_j - \sum_{t,r}b_{tr}c_{tr}.
\]
As in the fixed-capacity setting, we bound the optimal welfare by the optimal
fractional welfare; define
\[
    \fopt := \max_{x} \fval(x,\costs),
\]
where $x$ ranges over $x \ge 0$ satisfying the demand constraint
(Equation~\ref{eq:demand}).\footnote{Although not stated here in this
  manner, $\fopt$ can be expressed as the optimum of a linear program.}  We defer our proof of
Lemma~\ref{lem:fopt-upperbound-costs} to \Cref{sec:deferred}.
\begin{lemma}
    \label{lem:fopt-upperbound-costs}
    $\fopt_{\costs} \ge \opt_{\costs}$.
\end{lemma}

We now show that our analysis (theorems in
\Cref{sec:ub-interval} along with Lemma~\ref{lem:FGL}) extends to the
setting with convex production costs as well. We begin by updating our
definition of fractional unit allocations. Recall that unit
allocations implicitly define a partition of items into bundles and
associate each job with at most one bundle. When items have production
costs, the value of a fractional unit allocation depends on which copy
of each item belongs to any particular bundle in the partition. This
association of copies to bundles is not necessarily unique. So we
extend the definition of unit allocations to make the partitioning of
items into bundles explicit. 

\begin{definition} A {\em fractional unit allocation with costs} is a
  pair $(x,\tau)$, where $x$ is a demand-feasible fractional allocation and
  $\tau= \{\tau_1, \tau_2, \tau_3, \cdots\}$ is a partition of the multiset of
  items $\{(t,r)\}_{t\in T, r\in\Z^+}$, if there exists a partition of
  jobs $j\in U$ with $x_j>0$ into sets $\{A_1, A_2, A_3, \cdots\}$,
  such that:
    \begin{itemize}
    \item For all $k$, the set of items contained in $\tau_k$, $T_k=\{t:
      (t,r)\in \tau_k\}$, forms an interval.
        \item For all $j\in U$ with $x_j>0$, there is exactly one
          index $k$ with $j\in A_k$. 
        \item For all $k$ and $j\in A_k$, $I_j\subseteq T_k$.
        \item For all $k$, we have $\fwt(x_{A_k}) \le 1$.
% there exists an index $k_j$ with
%             $\intj\subseteq T_{k_j}$ and $\intj\cap T_{k'}=\emptyset$ for all
%             $k'\ne k_j$.
%         \item Denoting $A_k = \{j: k_j = k\}$, we have $\fwt(x_{A_k}) \le 1$.
    \end{itemize}
    The fractional value of $(x,\tau)$ is defined as:
    \[\fval(x,\tau,\costs) = \sum_j v_jx_j - \sum_k \sum_{(t,r)\in
        \tau_k}  \sum_{j\in A_k: t\in I_j} c_{tr}x_j.\]
\end{definition}

\noindent
We can now state a counterpart to Lemma~\ref{lem:FGL}, with the proof
deferred to \Cref{sec:deferred}.
\begin{lemma}
    \label{lem:FGL-costs}
    For any cost vector $\costs$ and a fractional unit allocation with
    costs, $(x,\tau)$, there exists an anonymous bundle pricing
    $\prices$ such that
    \[
        \sw(\prices) \ge \half \fval(x,\tau,\costs).
    \]
\end{lemma}

We emphasize that the pricing returned by the above lemma is not
necessarily static. For each bundle $\tau_k$ in the unit allocation
$(x,\tau)$, our pricing associates a price with each interval
$I\subseteq T_k$. Prices for different such subsets $I$ can be
different and depend on the costs of the respective items. When one of
the intervals in $\tau_k$ is bought, the seller removes from the menu
all other intervals corresponding to $\tau_k$. Moreover, different
copies of the same bundle can have different prices, because they cost
different amounts to the seller. We leave open the question of
obtaining a static bundle pricing with a good approximation factor for
this setting.

Next we state and prove a counterpart of
Theorem~\ref{thm:arbit-capacity}: given any fractional allocation $x$
we can construct a fractional unit allocation that captures an
$O(\log L/\log\log L)$ fraction of the value of $x$. The proof is
essentially a reduction to the fixed-capacity setting. We first
partition the allocation $x$ into ``layers'' as in the proof of
Theorem~\ref{thm:arbit-capacity}. The fractional weight for every item
in the allocation for any single layer is at most 1. In this case, it
becomes possible to apply the unit-capacity theorem
(Theorem~\ref{thm:unit-capacity}) to construct a
partition of items in that layer and a corresponding unit
allocation. Putting these layer-by-layer unit-allocations together
gives us an overall fractional unit allocation.

% Next, we show that Theorem~\ref{thm:arbit-capacity}, which asserts the
% existence of a good unit allocation, generalizes to the setting with costs.
% The proof is essentially the same as in the fixed-capacity setting: we first
% prove that if none of the items has fractional weight greater than 1, then it
% is possible to reduce this case to the setting of unit capacity without cost.
% Otherwise, we partition the fractional allocation into layers, and each job is
% assigned to one layer, while the total cost does not increase. 

% \begin{lemma}
%     \label{lem:unit-cost}
%     For all non-decreasing costs $\costs$ and every fractional allocation $x$
%     such that $\max_t b_t(x)\leq 1$, there exists a fractional unit allocation
%     $(x',\tau)$, such that
%     \[
%         \fval(x,\costs)\le  O(\log L/\log\log L)\fval(x',\tau,\costs).
%     \]
% \end{lemma}

% \begin{proof}
%     Notice that if $\max_t b_t(x)\leq 1$, then
%     \[
%         \fval(x,\costs) = \sum_j v_jx_j - \sum_{r}b_{tr}c_{tr} =
%                 \sum_j \Big(v_j-\sum_{t\in I_j}c_{t1}\Big)x_j.
%     \]
%     Thus the problem reduces to the unit-capacity setting, where the value of
%     each job $j$ is replaced by $v'_j=v_j-\sum_{t\in I_j}c_{t1}$. We
%     can then apply Theorem~\ref{thm:unit-capacity} to obtain a
%     fractional unit allocation. 
%     % The
%     % correctness of the lemma now follows immediately from
%     % Theorem~\ref{thm:unit-capacity}.
% \end{proof}

\begin{lemma}
    \label{lem:arbit-cost}
    For all non-decreasing costs $\costs$ and every fractional allocation $x$,
    there exists a fractional unit allocation $(x',\tau)$ such that
    \[
        \fval(x,\costs)\le  O(\log L/\log\log L)\fval(x',\tau,\costs).
    \]
\end{lemma}

\noindent
Lemmas~\ref{lem:fopt-upperbound-costs}, \ref{lem:FGL-costs} and
\ref{lem:arbit-cost} together imply the following upper bound for the
costs setting.

%\begin{numberedtheorem}{\ref{thm:costs-ub}}
\begin{theorem}
\label{thm:costs-ub}
For the interval preferences setting with increasing marginal costs on
items, there exists an anonymous bundle pricing $\prices$ such that
\[
    \opt \le O\left(\frac{\log L}{\log\log L}\right)\sw(\prices).
\]
\end{theorem}
%\end{numberedtheorem}

The lower bound of $\Omega\left(\frac{\log L}{\log\log L}\right)$ also extends
to the setting with costs, as the unit-capacity setting is a special case.

\section{Deferred proofs}
\label{sec:deferred}

\subsection{Proof of Lemma~\ref{lem:FGL}}

\begin{numberedlemma}{\ref{lem:FGL}}
   For any feasible fractional unit allocation $x$, there exists a static,
    anonymous bundle pricing $\prices$ such that
    \[
        \sw(\prices) \ge \half \fval(x).
    \]
\end{numberedlemma}
\begin{proof}
    Observe that we can write the social welfare of any $\prices$ as
    the sum of the expected revenue of the seller and the total
    expected utility obtained by the buyers:
    \[
        \sw(\price) = \rev(\price) + \util(\price).
    \]
    Let $\{T_k\}$ and $\{A_k\}$ be the partition of the items and jobs
    respectively corresponding to the unit allocation $x$. For each
    bundle $T_k$, let $W_k = \fwt(x_{A_k})$, and set price
    $\pricek = \frac1{2W_k}\fval(x_{A_k})$. For this setting of
    prices, we will bound the revenue and utility terms separately,
    and show that their sum is at least
    $\half\fval(x) = \half\sum_{k}\fval(x_{A_k})$.

    For any realization of the buyers' values and arrival order, let $Z_k$
    indicate whether bundle $k$ is purchased by {\em any} buyer at these
    prices.  The revenue term is then
    \begin{align}
        \rev(\price)  &= \sum_{k}\Pr[Z_k=1]\pricek \nonumber \\
            &\ge \half\sum_{k}\Pr[Z_k=1]\fval(x_{A_k}). \label{eq:rev}
    \end{align}
    The inequality follows from the fact that $W_k \le 1$ by the
    definition of unit allocations.
    
    For the utility, note that buyer $i$ with value $\vali$ will purchase an
    available bundle maximizing her utility. Define $k_j$ to be the index $k$
    such that $j \in A_k$. The buyer's utility is at
    least\footnote{Here we use the notation $y^+$ to denote $\max\{0,y\}$.}
    \begin{align}
        u_i(\vali) &\ge \expect[Z]{\max_{j\in\Jvi} \indicate\{Z_{k_j}=0\}
                \big(\vali(\intj) - \price[k_j]\big)^+} \nonumber \\
            &\ge \frac1{\qvi}\sum_{j\in\Jvi} \Pr[Z_{k_j}=0]\xj(\vj -
                \price[k_j])^+. \label{eq:vali-util}
    \end{align}
    The second inequality follows from \eqref{eq:demand}.

    Summing over all buyers and all valuations, we have
    \begin{align}
        \util(\price)  &\ge \sum_{i,\vali}\qvi u_i(\vali) \nonumber \\
              &\ge \sum_{i,\vali}\sum_{j\in\Jvi}\Pr[Z_{k_j}=0]\xj(\vj -
                  \price[k_j])^+ \nonumber \\
              &\ge \sum_k \Pr[Z_k=0]\sum_{j\in A_k} \xj(\vj - \pricek) 
                  \nonumber \\
              &\ge \sum_k \Pr[Z_k=0]
                \Big(\fval(x_{A_k})-W_k\pricek  \Big) 
                  \nonumber \\
              &\ge \sum_k \Pr[Z_k=0]
                \Big(\fval(x_{A_k})-\frac 12 \fval(x_{A_k})\Big) 
                  \nonumber \\
              &= \half \sum_k \Pr[Z_k=0] \fval(x_{A_k}) \label{eq:util}
    \end{align}

    The result now follows by summing \eqref{eq:rev} and \eqref{eq:util}.
\end{proof}

\subsection{Proof of Lemma~\ref{lem:small-v-bound}}

\begin{numberedlemma}{\ref{lem:small-v-bound}}
For any fractional allocation $x$, there exists a set of jobs $U_1$ as defined
in \Cref{sec:unit-cap-ub} such that $\fval(x_{U_1})\ge \half\fval(x)$.
\end{numberedlemma}
\begin{proof}
We bound the total fractional value of jobs left out of $U_1$:
    \begin{align*}
        \sum_{j\not\in U_1} v_jx_j &< \sum_{j\not\in U_1} 
                \left(\half\sum_{t\in I_j}\fvt\right) x_j \\
            &= \half \sum_{j\not\in U_1} x_j \left(\sum_{t\in I_j}
                \sum_{j' : I_{j'} \ni t} \densj[j'] x_{j'} \right) \\
            &= \half \sum_t\left(\sum_{\substack{j\not\in U_1\\I_j\ni t}}x_j\right)
                \sum_{j' : I_{j'} \ni t} \densj[j'] x_{j'} \\
            &\le \half \sum_t\sum_{j : I_j \ni t} \densj x_j = \half \fval(x).
    \end{align*}
\end{proof}

\subsection{Proof of Lemma~\ref{lem:a-cell-bound}}

\begin{numberedlemma}{\ref{lem:a-cell-bound}}
    For each interval $\Int_{\ell,k}$, there exists a set of jobs
    $S_{\ell,k} \subseteq \Gtilde_{\ell,k}$ such that
    \begin{enumerate}[label=\roman*., leftmargin=2\parindent]
        \item $\fwt(x, S_{\ell,k}) \le \frac1\beta$ and
        \item $\fval(x, S_{\ell,k}) \ge \frac16 \fval(x, \Gtilde_{\ell,k})$.
    \end{enumerate}

\end{numberedlemma}

\begin{proof}
    Since all jobs in $\Gtilde_{\ell,k,a}$ have value between $2^{a-1}$ and $2^a$, we have
    \[
        2^{a-1}\fwt(x, \Gtilde_{\ell,k,a}) \leq\fval(x,\Gtilde_{\ell,k,a}) \leq
            2^a\fwt(x, \Gtilde_{\ell,k,a}).
    \]
    For every $u\in\{1,\cdots,\lceil\log \vmax\rceil\}$, define
    $S_u=\sum_{a\leq u}2^a\fwt(x,\Gtilde_{\ell,k,a})$.  Then $S_{\lceil\log
    \vmax\rceil}$ is an upperbound of $\fval(x, \Gtilde_{\ell,k})$.  Let $m$ be such that $S_m\leq
    \frac{1}{3}S_{\lceil\log \vmax\rceil}$ and
    $S_{m+1}>\frac{1}{3}S_{\lceil\log \vmax\rceil}$.  Consider the following
    two cases.

    \medskip
    {\em Case 1.} If $S_{m+1}-S_m = 2^{m+1}\fwt(x,\Gtilde_{\ell,k,m+1}) >
    \frac{1}{3}S_{\lceil\log \vmax\rceil}$, then
    \begin{align*}
        \fval(x,\Gtilde_{\ell,k,m+1}) &\geq 2^{m}\fwt(x,\Gtilde_{\ell,k,m+1}) \\
            &> \frac{1}{6}S_{\lceil\log \vmax\rceil} \\
            &\geq\frac16 \fval(x,\Gtilde_{\ell,k}).
    \end{align*}
    Since $\fwt(x, \Gtilde_{\ell,k,m+1})\leq \frac{1}{2\beta}<\frac{1}{\beta}$,
    setting $S_{\ell,k}=\Gtilde_{\ell,k,m+1}$ satisfies both conditions
    of the lemma.

    \medskip
    {\em Case 2.} If $S_{m+1}-S_m\leq \frac{1}{3}S_{\lceil\log \vmax\rceil}$,
    then $S_{m+1}\leq\frac{2}{3}S_{\lceil\log \vmax\rceil}$, and
    \begin{align*}
        \sum_{a>m+1}\fval(x,\Gtilde_{\ell,k,a}) &\geq
                \frac{1}{2}(S_{\lceil\log \vmax\rceil}-S_{m+1}) \\ 
        &\geq \frac{1}{6}S_{\lceil\log\vmax\rceil} \\
        &\geq \frac16\fval(x,\Gtilde_{\ell,k}).
    \end{align*}
    Meanwhile, note that
    $\sum_{a>m+1}2^a\fwt(x,\Gtilde_{\ell,k,a}) < \frac{2}{3}S_{\lceil\log
    \vmax\rceil} < 2S_{m+1}.$ Then
    \begin{align*}
        \sum_{a>m+1}\fwt(x,\Gtilde_{\ell,k,a}) &\leq
                \frac{1}{2^{m+2}}\sum_{a>m+1} 2^a\fwt(x,\Gtilde_{\ell,k,a}) \\
        &< \frac{1}{2^{m+2}}\cdot2S_{m+1} \\
        &= \frac{1}{2^{u+1}}\sum_{a\leq m+1}2^a\fwt(x,\Gtilde_{\ell,k,a}) \\
        &\leq \frac{1}{2^{u+1}}\sum_{a\leq m+1}2^a\frac{1}{2\beta} < \frac{1}{\beta}.
    \end{align*}
    Thus setting $S_{\ell,k}=\Union_{a>m+1}\Gtilde_{\ell,k,a}$ satisfies
    both conditions of the lemma.
\end{proof}

\subsection{Proof of Lemma~\ref{lem:light-bound}}

\begin{numberedlemma}{\ref{lem:light-bound}}
    There exists a fractional unit allocation $\tilde{x}_\light$ such that 
    \[
        \sum_{G\in\light} \fval(x_G) \le
        O\left(\frac{\log L}{\log\beta}\right)\fval(\tilde{x}_\light).
    \]
\end{numberedlemma}

\begin{proof}
  Fix any length scale $\ell$ and consider the partition of items into
  intervals $\Int_{\ell,k}$. We will construct a fractional unit
  allocation corresponding to this partition. We first associated with
  the interval $\Int_{\ell,k}$ the set of all jobs contained inside
  this interval that have length scale between $\ell-\log\beta+1$ and
  $\ell$. Formally, we define
    \begin{equation*}
        H_{\ell,k}=\Union_{\substack{ \ell',k':
            \Int_{\ell',k'}\subseteq \Int_{\ell,k} \\ \ell'\in
            (\ell-\log\beta, \ell]}} S_{\ell',k'},
    \end{equation*}
   where $S_{\ell',k'}$ is defined in
   Lemma~\ref{lem:a-cell-bound}. See Figure~\ref{fig:intervalblock} for 
   a graphical illustration of $H_{\ell,k}$.
   Let $H_\ell = \cup_k H_{\ell,k}$
   denote the set of all jobs that belong to the partition given by
   $\{H_{\ell,k}\}$. 

   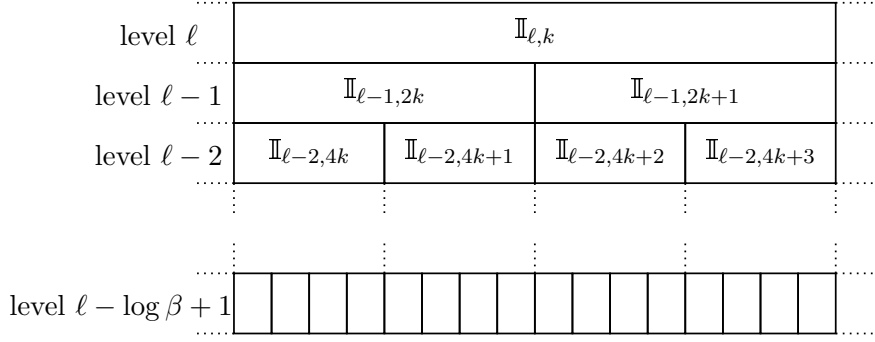
\begin{figure}
       \centering
    \newcommand{\cellwidth}{0.5 cm}
    \newcommand{\cellheight}{0.8 cm}

    \begin{tikzpicture}
        [x=\cellwidth,y=-\cellheight,node distance=0 cm,outer sep=0 pt,line
        width=0.66pt]

        \tikzstyle{cell}=[rectangle,draw,
            minimum height=\cellheight,
            anchor=north west,
            text centered]
        \tikzstyle{w16}=[cell,minimum width=16*\cellwidth]
        \tikzstyle{w8}=[cell,minimum width=8*\cellwidth]
        \tikzstyle{w4}=[cell,minimum width=4*\cellwidth]
        \tikzstyle{w2}=[cell,minimum width=2*\cellwidth]
        \tikzstyle{w1}=[cell,minimum width=1*\cellwidth]
        \tikzstyle{job}=[rectangle,draw,fill=gray!30,anchor=north west]

        % top rows
        \node[anchor=north] at (-2, 0.2) {level $\ell$};
        \node[anchor=north] at (-2, 1.2) {level $\ell-1$};
        \node[anchor=north] at (-2, 2.2) {level $\ell-2$};
        \node[anchor=north] at (-3, 4.7) {level $\ell-\log \beta+1$};
        \node[w16] at (0, 0) {$\Int_{\ell,k}$};
        \node[w8]  at (0, 1) {$\Int_{\ell-1,2k}$};
        \node[w8]  at (8, 1) {$\Int_{\ell-1,2k+1}$};
        \node[w4]  at (0, 2) {$\Int_{\ell-2, 4k}$};
        \foreach \k in {1,2,3} {
            \node[w4] at (4*\k, 2) {$\Int_{\ell-2, 4k+\k}$};
        }

        % horizontal continuation...
        \foreach \y in {0,1,2,3, 4.5,5.5} {
            \draw[dotted] (0, \y) -- (-1,\y);
            \draw[dotted] (16,\y) -- (17,\y);
        }

        % gap
        \foreach \x in {0,4,...,16} {
            \draw[dotted] (\x,3) -- (\x,3.5);
            \draw[dotted] (\x,4.5) -- (\x,4);
        }

        % bottom row
        \foreach \x in {0,...,15} {
            \node[w1] at (\x, 4.5) {};
        }

    \end{tikzpicture}
    \caption{Intervals associated with $H_{\ell,k}$ in
       Lemma~\ref{lem:light-bound}. There are $\beta-1$ intervals in the block,
       and from each interval $\Int_{\ell',k'}\subseteq \Int_{\ell,k}$ we pick
       a set of jobs $S_{\ell',k'}$ with fractional weight at most $1/\beta$.}
    \label{fig:intervalblock}
\end{figure}

   We now claim that $x_{H_\ell}$ is a fractional unit allocation. To
   see this, observe that for each $(\ell,k)$ there are exactly
   $2^{\log\beta}-1<\beta$ pairs $(\ell',k')$ with
   $\Int_{\ell',k'}\subseteq \Int_{\ell,k}$ and
   $\ell'\in (\ell-\log\beta, \ell]$. $H_{\ell,k}$ is therefore a
   union of at most $\beta$ groups
   $S_{\ell',k'}$. Lemma~\ref{lem:a-cell-bound} now implies that
   $\fwt(x, H_{\ell,k})\le \beta \frac 1\beta = 1$.

   Lemma~\ref{lem:a-cell-bound} also implies that the total fractional
   value captured by $x_{H_\ell}$ is a constant fraction of the total
   value of all light weight groups at length scales in
   $(\ell-\log\beta, \ell]$: 
   \begin{align*}
     \fval(x,H_\ell) = \sum_k \fval(x,H_{\ell,k}) & = \sum_k 
       \sum_{\ell'\in (\ell-\log\beta, \ell]} \fval(x, S_{\ell',k})\\
     & \ge \frac 16 \sum_k 
       \sum_{\ell'\in (\ell-\log\beta, \ell]}\fval(x, \Gtilde_{\ell',k})\\
     & = \frac 16 \sum_{\substack{G\in \light;\\ G \text{ at length scale } \ell'\in (\ell-\log\beta, \ell]}}\fval(x_G)
   \end{align*}

   Now consider the $\lceil\lmax/\log\beta\rceil$ length scales in $\{1, \cdots,
   \lmax\}$ that are multiples of $\log \beta$. By our argument above,
   the fractional allocations $x_{H_\ell}$ corresponding to such
   length scales together capture a sixth of all of the
   fractional value in $\light$. Therefore, there exists some $\ell$
   such that 
   \[\fval(x,H_\ell)\ge \Omega\left(\frac
     {\log\beta}{\lmax}\right)\sum_{G\in\light} \fval(x_G).\]
   The corresponding unit allocation $x_{H_\ell}$ satisfies the
   requirements of the lemma.
\end{proof}

\subsection{Proof of Lemma~\ref{lem:greedy-layer}}

\begin{numberedlemma}{\ref{lem:greedy-layer}}
    For any feasible fractional allocation $x$ in the interval preferences
    setting with arbitrary capacities, one can efficiently
    construct a set $S$ of jobs such that the total fractional weight of
    $x_S$ at any item $t$ is at least $\min\{1,B_t\}$ and at most $4$.
    Formally, for all items $t$, $\min\{1,B_t\}\le \sum_{j\in S: t\in
    \intj} x_j < 4$.
\end{numberedlemma}

\begin{proof}
    We build up the set $S$ greedily as follows. On each iteration,
    starting from the earliest item $t$ which is not fully covered, we add
    the job whose interval contains $t$ and ends the latest. More formally,
    let $S_0 = \emptyset$. For any set $S$, let $(x,S)|_t := \sum_{j \in S
    : \intj \ni t} x_j$. On iteration $i$, find $t_i = \min \{t:
    (x,S_{i-1})|_t < \min(1, \capt)\}$. Find $j_i = \argmax\{t \in \intj :
    j \not\in S_{i-1}, t_i \in \intj\}$, and set $S_i = S_{i-1} \union
    \{j_i\}$.  Let $S = S_{i^*}$, where $i^*$ is the last iteration.

    By construction, $\min\{1,\capt\}\le\sum_{j\in S: t\in\intj}x_j$ for all
    $t$.  Suppose, by way of contradiction, that there exists $t'$ such that
    $\sum_{j\in S: t'\in\intj}x_j \ge 4$. Let $\{j_1,j_2,\cdots\} \subseteq S$
    be all jobs in $S$ containing $t'$, indexed in the order in which they were
    added to $S$. Let $a$ be such that $1 \le \sum_{i=1}^ax_{j_i} < 2$, and let
    $U_< = \{j_1,\cdots,j_a\}$. Similarly, sort the jobs such that they are
    indexed in order by decreasing end time; let $b$  be such that $1 \le
    \sum_{i=1}^bx_{j_i} < 2$, and let $U_> = \{j_1,\cdots,j_b\}$. Note that
    $\sum_{j\in U_< \union U_>} x_j < 4$, so there exists a job $j'$ which is
    not in either $U_<$ or $U_>$. This job was added to $S$ after all jobs in
    $U_<$ were added, and it ends before all jobs in $U_>$. Let $i'$ be the
    iteration on which $j'$ was added to $S$; we will show there is no item
    $t_{i'}$ for which the algorithm would have added $j'$.
    
    Observe that, because the algorithm covers items starting with the earliest
    uncovered item first, at every iteration $i$ the set of covered items
    (i.e., the set $\{t : (x,S_{i-1})|_t \ge \min(1,\capt)\}$) is a prefix of
    $T$, the set of all items. Therefore, by the time the algorithm has added
    $U_<$ to $S$, all items up to and including $t'$ have been covered. So
    $t_{i'}$ cannot come before or be equal to $t'$.  But all jobs in $U_>$
    contain item $t'$, and therefore must have been added to $S$ before job
    $j'$ because they end later. So item $t_{i'}$ cannot come after $t'$.
    Therefore job $j'$ cannot have been added to $S$.
\end{proof}

\subsection{Proof of Lemma~\ref{lem:peeling}}

\begin{numberedlemma}{\ref{lem:peeling}}
    Given a set $U$ of jobs and a non-zero fractional allocation $y$,
    there exist sets of arms $\tlayer \ne \emptyset$ and $D$ such that 
    \begin{enumerate}[label=\roman*.]
        \item $\tlayer \intersect D = \emptyset$;
        \item $\tlayer \intersect D' = \emptyset$, where $D' = \{\armj : j^{a'}
            \in D, a\ne a'\}$;
        %\item $y'_{\tlayer \union D} = 0$,
        %\item $y'$ is feasible for the multiset of items $T \setminus T_l$,
        \item for all $t \in E_\tlayer$, $\fwe(y, \tlayer \union D)\ge
          \min\{1, \fwe(y)\}$; \label{item:lowerbound}
          % $w_t(\tlayer \union D, y) \ge \min\{1, w_t(U, y)\}$
        \item for all $t \in E_\tlayer$, $\fwe(y, \tlayer) \le 7$; and
          % $w_t(\tlayer, y) \le 7$; and
            \label{item:upperbound}
        \item $\fwt(y_\tlayer) \ge 2 \fwt(y_D)$.
    \end{enumerate}
    % where $T_\tlayer = \{t : w_t(\tlayer, y) > 0\}$. 
    Furthermore, $\tlayer$ and $D$ can be found efficiently.
   % Given a set of jobs $U$ and a non-zero feasible fractional allocation $y$,
   %  there exist sets of arms $\tlayer \ne \emptyset$ and $D$ such that 
   %  \begin{enumerate}[label=\roman*.]
   %      \item $\tlayer \intersect D = \emptyset$,
   %      %\item $y'_{\tlayer \union D} = 0$,
   %      %\item $y'$ is feasible for the multiset of items $T \setminus T_l$,
   %      \item for all $t \in T_\tlayer$, $w_t(\tlayer \union D, y) \ge \min\{1, w_t(U, y)\}$,
   %      \item for all $t \in T_\tlayer$, $w_t(\tlayer, y) \le 7$, and
   %      \item $\fwt(y_\tlayer) \ge 2 \fwt(y_D)$,
   %  \end{enumerate}
   %  where $T_\tlayer = \{t : w_t(\tlayer, y) > 0\}$. Furthermore, $\tlayer$ and $D$
   %  can be found efficiently.
\end{numberedlemma}

\begin{proof}

    If $\fwe(y) \le 6$ for all $t \in T$, then set $\tlayer = U$ and $D =
    \emptyset$. This satisfies the conditions. So assume that
    there exists $t$ such that $\fwe(y) > 6$. Pick such a $t$ with maximal
    depth $d(t)$. Thus, for all $t'$ in the subtree below $t$, $\fwe[t'](y)
    \le 6$.

    Suppose there is a set of arms $S$ which begin at or above $t$ and
    terminate at $t$ such that $\fwe(y, S) > 3$. Define the {\em depth} of an arm
    to be the depth of the highest edge in its desired interval, so that arms
    with low depth desire intervals beginning higher in the tree. Then let
    $\armj[j_1], \armj[j_2], \cdots, \armj[j_k], \cdots, \armj[j_\ell]$ be a
    subset of $S$ sorted in decreasing order by depth such that
    \[
        \sum_{i=1}^{\ell-1}y_{\armj[j_i]} < 3 \le \sum_{i=1}^\ell y_{\armj[j_i]}
    \]
    and
    \[
        \sum_{i=k+1}^{\ell} y_{\armj[j_i]} < 1 \le \sum_{i=k}^\ell y_{\armj[j_i]}.
    \]
    Then let $\tlayer = \{\armj[j_1], \cdots, \armj[j_k]\}$ and $D =
    \{\armj[j_{k+1}], \cdots, \armj[j_\ell]\}$.  Because we are considering
    arms which terminate at $t$, $\intj[{\armj[j_i]}] \subseteq
    \intj[{\armj[j_{i+1}]}]$ for all $i \in [\ell-1]$. In particular,
    $\fwe[e](y, \tlayer
    \union D)\geq \sum_{i=k}^{\ell}y_{\armj[j_i]}\ge 1$ for all $e \in T_\tlayer$.
    Observe that $\fwt(y_\tlayer) \ge 2 \fwt(y_D)$. %so let $y'$ be $y$ with
    %all allocations to arms in $\tlayer$ and $D$ zeroed out. Formally, set $A =
    %\tlayer\union D$, and set $y' = y_{U\setminus A}$. 
    This satisfies the conditions of the lemma.

    Suppose there is no set $S$ of arms terminating at $t$ such that
    $\fwe(y, S)
    > 3$. Then there must exist a set of arms including $t$ but terminating
    below $t$ with total weight at least $3$. Let $t_1, \cdots, t_d$ be the $d$
    edges below $t$.  Let $S_i$ be the set of arms that contain both $t$ and
    $t_i$, and $T_i$ be the set of arms containing only edges within the
    subtree planted at $t_i$. Because $\fwe[t_i](y, S_i) \le 6$ for all $i$ by our
    choice of $t$, we can find $m$ such that $3 \le \sum_{i=1}^m
    \fwe(y, S_i) <
    9$. Let $S = \Union_{i=1}^m S_i$. Let $c = \fwe(y, S)$, so $3 \le c < 9$.
    Let $\armj[j_1], \cdots, \armj[j_k], \cdots, \armj[j_\ell]$ be the arms in
    $S$ sorted in decreasing order by depth, as before, where 
    \[
        2\sum_{i=k+1}^{\ell} y_{\armj[j_i]} \le \sum_{i=1}^k y_{\armj[j_i]}\leq 7
    \]
    and
    \[
        \sum_{i=k}^{\ell} y_{\armj[j_i]} \geq 1.
    \]

    \begin{figure}[tbp]
    \centering
        \subfigure{
        \begin{minipage}[b]{1.1\textwidth}
            \includegraphics[width=0.33\textwidth]{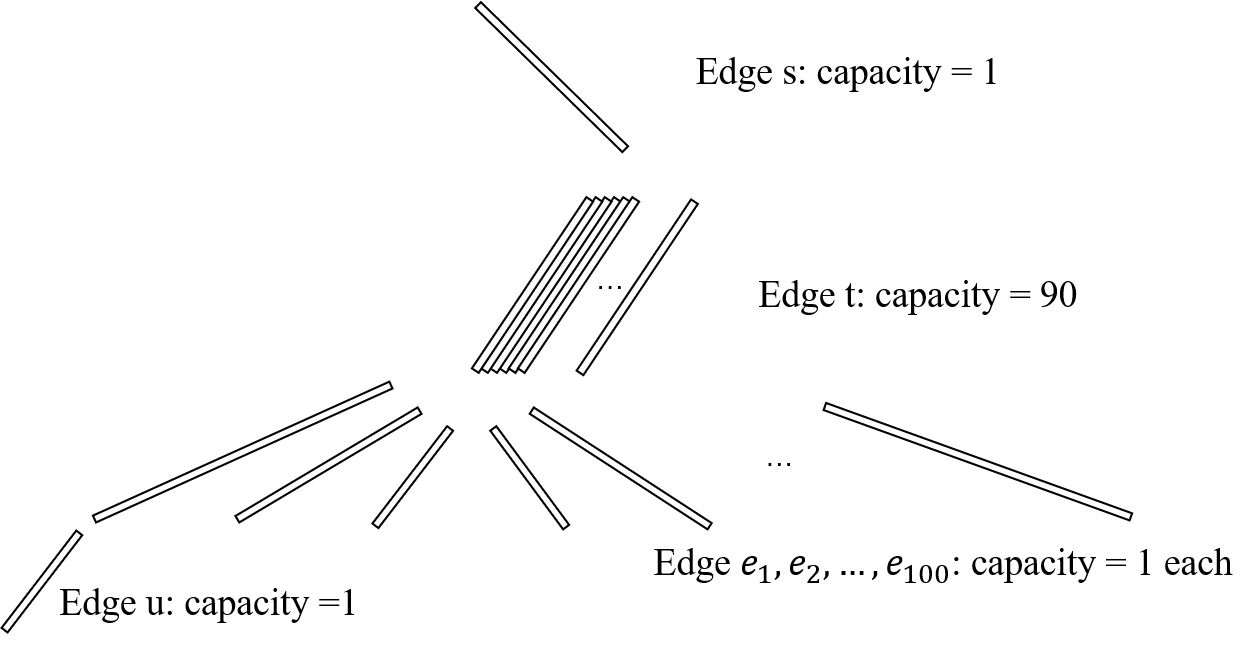}
            \includegraphics[width=0.33\textwidth]{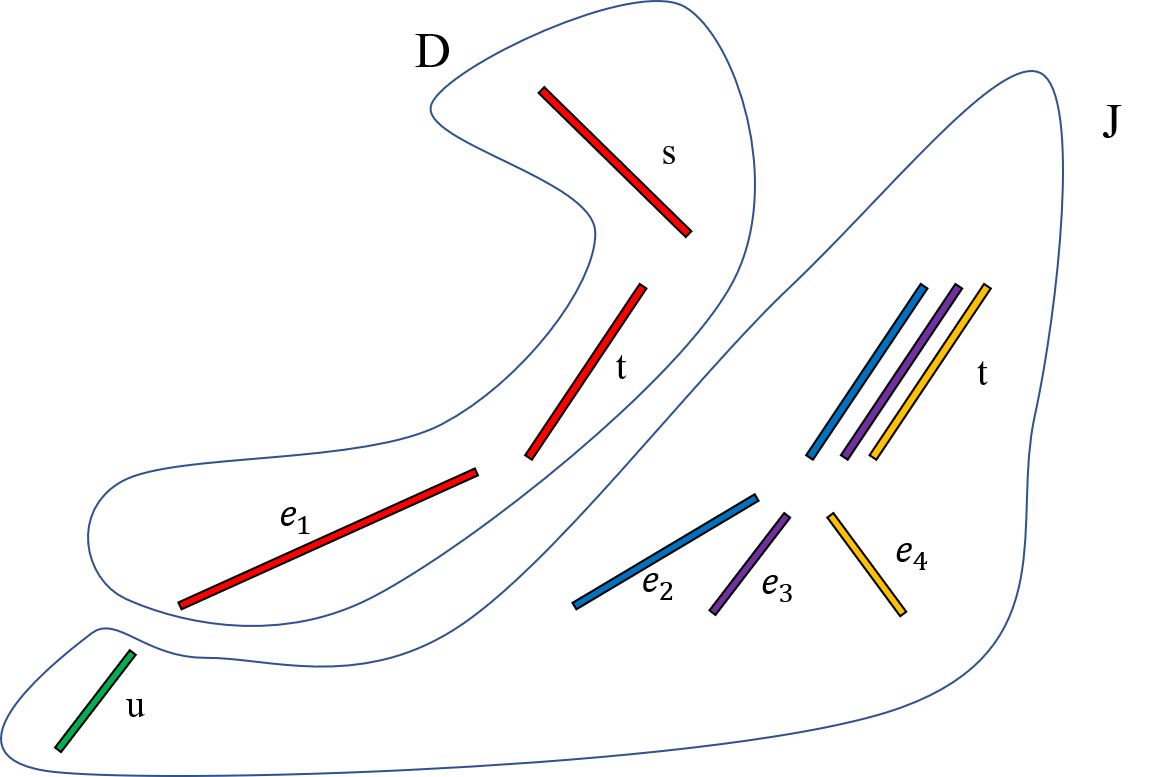}
            \includegraphics[width=0.33\textwidth]{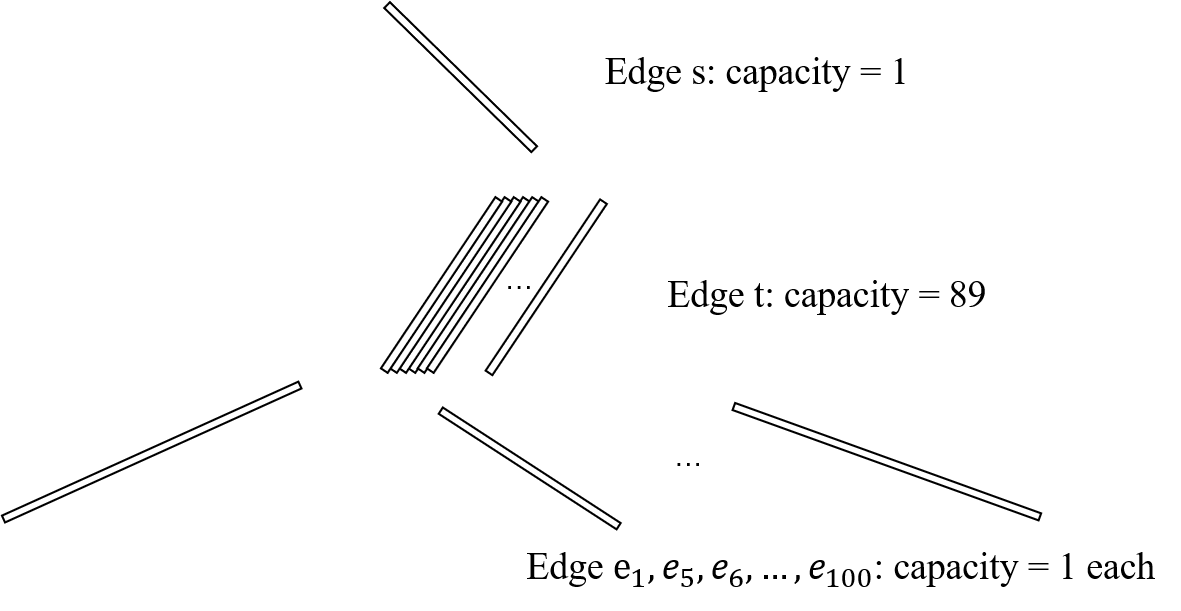}
 %           \caption{\label{fig:lb-example}\small{An example of a binary tree and its
 %             corresponding grid.}}
        \end{minipage}
        }
        \caption{\small{An example for one step of layering
            algorithm. (Left figure:) At the beginning of the layering
            step, there are 101 jobs: $(s-t-e_1)$, $(t-e_2)$,
            $(t-e_3)$, $\cdots$, $(t-e_{100})$, $(u)$. Each job has
            value 1 and allocation 0.9. The capacity of edge t is 90,
            while the capacity of other edges are 1. (Middle figure:)
            In the first step of the layering process, job $(s-t-e_1)$
            is in set $D$, jobs $(t-e_2)$, $(t-e_3)$, $(t-e_3)$ and
            $(u)$ are in set $\tlayer$. The job in set $D$ is deleted,
            while the jobs in set $\tlayer$ have their weight scaled
            down and allocated to one layer. (Right figure:) The
            capacity of edges after the first step of layering
            procedure. The capacity of edges that have positive weight
            in set $\tlayer$ are decreased by 1.}}
    \end{figure}
 
    Then set $\tlayer = \{\armj[j_i] : i = 1, \cdots, k\} \union T_1 \union
    \cdots \union T_m$ and $D = \{\armj[j_i] : i = k+1, \cdots, \ell\}$. For
    each edge $e\in T_\tlayer$, $\fwe(y, \tlayer)\leq 7$. The same as in
    previous case for each edge $e$ that is not below $t$, $\fwe[e](y,
    \tlayer \union
    D)\geq \sum_{i=k}^{\ell}y_{\armj[j_i]}\ge 1$; for each edge $e$ below $t$
    in $T_\tlayer$, $\fwe[e](y, \tlayer \union D)=\fwe(y)$. $\fwt(y_\tlayer) \ge
    \sum_{i=1}^k y_{\armj[j_i]}\ge 2\sum_{i=k+1}^{\ell}
    y_{\armj[j_i]}=\fwt(y_D)$.  Thus every condition of the lemma is satisfied.
\end{proof}

\subsection{Proof of Lemma~\ref{lem:fopt-upperbound-costs}}

\begin{numberedlemma}{\ref{lem:fopt-upperbound-costs}}
    $\fopt_{\costs} \ge \opt_{\costs}$.
\end{numberedlemma}
\begin{proof}
    Let $x^*_j$ be the probability that job $j$ is allocated in optimal offline allocation, where
    the randomness is over all possible realizations of the values. Then $x^*$ is a fractional
    allocation of jobs.
    Let $d_{tr}$ be the probability that the $r$th copy of item $t$ is allocated in 
    offline optimal allocation. Then $0\leq d_{tr}\leq 1$, $b_t(x^*)=\sum_{r}d_{tr}$,
    \begin{equation*}
        \opt_{\costs}=\sum_{j}v_jx^*_j-\sum_{t,r}d_{tr}c_{tr}.
    \end{equation*}
    To prove the lemma, it suffices to show that $\fval(x^*,\costs)\geq\opt_{\costs}$, which is equivalent 
    to proving $\sum_{t,r}b_{tr}c_{tr}\leq\sum_{t,r}d_{tr}c_{tr}$. This is true since for every $t$ and $k$,
    $\sum_{r\geq k}d_{tr}\geq\sum_{r\geq k}b_{tr}$, and 
    \begin{eqnarray*}
        \sum_{r\geq 1}d_{tr}c_{tr}&=&c_{t1}\sum_{r\geq 1}d_{tr}+(c_{t2}-c_{t1})\sum_{r\geq 2}d_{tr}+(c_{t3}-c_{t2})\sum_{r\geq 3}d_{tr}+\cdots\\
        &\geq&c_{t1}\sum_{r\geq 1}b_{tr}+(c_{t2}-c_{t1})\sum_{r\geq 2}b_{tr}+(c_{t3}-c_{t2})\sum_{r\geq 3}b_{tr}+\cdots=\sum_{r\geq 1}b_{tr}c_{tr}.
    \end{eqnarray*} 
\end{proof}

\subsection{Proof of Lemma~\ref{lem:FGL-costs}}

\begin{numberedlemma}{\ref{lem:FGL-costs}}
    For any cost vector $\costs$ and a fractional unit allocation with
    costs, $(x,\tau)$, there exists an anonymous bundle pricing
    $\prices$ such that
    \[
        \sw(\prices) \ge \half \fval(x,\tau,\costs).
    \]
\end{numberedlemma}
\begin{proof}
    We will give a reduction to the fixed-capacity setting without costs and
    apply Lemma~\ref{lem:FGL}.  Recall that in a fractional unit allocation
    with costs, each job $j$ with $x_j > 0$ is explicitly associated with a
    particular bundle, and each bundle is explicitly associated with a
    particular copy of each item contained in the bundle.  Thus, the allocation
    defines a unique charging of the total cost to individual jobs. We
    construct a pricing instance in the fixed-capacity setting by subtracting
    from each job's value the cost of serving it. That is, for every job $j$ in
    the instance with costs, we construct a job $j'$ in the instance without
    costs with value 
    \begin{equation*}
        v_{j'}=v_j-\sum_{t\in I, (t,r)\in \tau_k}c_{tr},
    \end{equation*}
    where $k$ is such that $j \in A_k$ (i.e., the fractional unit allocation
    with costs assigns job $j$ to bundle $k$).  Thus the two instances have
    identical fractional values, i.e.  $\fval'(x)=\fval(x,\tau,\costs)$. By
    Lemma~\ref{lem:FGL}, there exists a static, anonymous bundle pricing
    $\prices'$ such that
    \begin{equation*}
        \sw'(\prices') \ge \half \fval'(x),
    \end{equation*}
    and each set $\tau_k$ is assigned a price $\price'_k$. We use this bundle
    pricing to associate the following prices $\prices$ to the original
    instance with costs. For each bundle $\tau_k$, for each interval $I
    \subseteq T_k$, set price
    \[
        \price_{\tau_k,I} = \price'_k + \sum_{t\in I, (t,r)\in \tau_k} c_{tr}.
    \]
    We refer to the price $\price'_k$ as the base price for bundle $\tau_k$,
    and the sum over costs as the surcharge for interval $I$.
    % For every interval $I$, its price is
    % $\min_{k:T_k\supseteq I}\prices'(\tau_k)+\sum_{t\in I, (t,r)\in
    % \tau_k}c_{tr}$ in the beginning.

    Our menu lists a price for every subinterval of every bundle, covering
    every copy of every item.\footnote{In practice, of course, for each
    interval the menu can list just the minimum price over all unsold bundles.}
    However, the mechanism sells at most one interval per bundle: when a buyer
    selects her preferred interval, the mechanism removes all intervals
    corresponding to the same bundle from the menu. Thus, our menu is adaptive
    (prices of unsold intervals change in response to demand for other
    intervals), though it remains anonymous.
    % Our adaptive bundle pricing runs as follows. Say a set $\tau_k$ is
    % purchased by job $(j,v_j,I_j)$ if when $j$ comes, it is allocated and
    % charged price $\prices'(\tau_k)+\sum_{t\in I_j, (t,r)\in \tau_k}c_{tr}$.
    % Once a set is purchased, the price of each interval $I$ is updated to 
    % \begin{equation*}
    %     \prices(I) = \min_{k:T_k\supseteq I, \tau_k\textrm{ unsold}} \prices'(\tau_k)+\sum_{t\in I, (t,r)\in \tau_k}c_{tr}.
    % \end{equation*}

    We now claim that the social welfare generated by pricing $\prices$ in the
    instance with costs is precisely equal to the social welfare generated by
    pricing $\prices'$ for the new fixed-capacity instance. We prove this
    pointwise by coupling the buyers' values and arrival order in the two
    settings, and then by induction over the number of buyers that have already
    arrived. Our inductive hypothesis is that at any point of time, the set of
    bundles available in the new instance is exactly the same as the set of
    unsold bundles in the original instance.  By the way we
    construct buyers' values in the new instance, buyers' preferences over
    bundles in the two instances are exactly coupled. Thus, the buyers' total
    utility is the same in each instance. Furthermore, the revenue generated by
    the base price of each bundle is equal to the revenue generated in the
    fixed-capacity instance, and the surcharge covers the production costs.

    This concludes the proof of the lemma.
\end{proof}

\subsection{Proof of Lemma~\ref{lem:arbit-cost}}

\begin{numberedlemma}{\ref{lem:arbit-cost}}
    For all non-decreasing costs $\costs$ and every fractional allocation $x$,
    there exists a fractional unit allocation $(x',\tau)$ such that
    \[
        \fval(x,\costs)\le  O(\log L/\log\log L)\fval(x',\tau,\costs).
    \]
\end{numberedlemma}

\begin{proof}
  We partition the multiset of items as well as the fractional
  allocation $x$ into layers using the procedure described in the
  proof of Theorem~\ref{thm:arbit-capacity} in
  Section~\ref{sec:arbit-cap-ub}. Let $Y_r$ denote the $r$th layer and
  $x^{(r)}$ denote the unscaled fractional allocation associated with
  layer $r$, so that $\sum_r x^{(r)} = x$. Let $S^{(r)}$ denote the
  set of jobs associated with $x^{(r)}$, so that
  $x^{(r)} = x_{S^{(r)}}$. Observe that by the manner in which we
  create layers, the sets $S^{(r)}$ partition the set of all jobs.

  Let
  $b'_{tr} = \sum_{j\in S^{(r)}: \intj \ni t}x_j = \sum_{j: \intj \ni
    t} x^{(r)}_j$ be the total fractional weight of jobs containing
  item $t$ in level $r$. Let $\hat{r}_t$ be the maximum layer index
  $r$ such that $b'_{tr}>0$. By the nature of the layering process,
  every layer below $\hat{r}_t$ is ``filled'': $1 \le b'_{tr} < 4$ for
  every $r < \hat{r}_t$. Then
  $\sum_{r\geq 1}b'_{tr}=\sum_{r\geq 1}b_{tr}=b_t$, and
  $\sum_{r\geq k}b'_{tr}\le \sum_{r\geq k}b_{tr}$ for every $k>1$,
  where the latter inequality follows by recalling that $b_{tr}\le 1$
  and $b'_{tr}\ge 1$ for all $r$. If we associate cost $c_{tr}$ with
  item $t$ in all jobs assigned to layer $r$, then the total cost of
  item $t$ is
\begin{eqnarray*}
\sum_{r\geq 1}b'_{tr}c_{tr}&=&c_{t1}\sum_{r\geq 1}b'_{tr}+(c_{t2}-c_{t1})\sum_{r\geq 2}b'_{tr}+(c_{t3}-c_{t2})\sum_{r\geq 3}b'_{tr}+\cdots\\
&\leq&c_{t1}\sum_{r\geq 1}b_{tr}+(c_{t2}-c_{t1})\sum_{r\geq 2}b_{tr}+(c_{t3}-c_{t2})\sum_{r\geq 3}b_{tr}+\cdots=\sum_{r\geq 1}b_{tr}c_{tr},
\end{eqnarray*} 
    with the last term being exactly the cost of item $t$ in
    fractional allocation $x$. 

    Now consider the fractional allocation $\xr$ for each layer
    $\layer$, defined as $\xr=x^{(r)}/4$. This is a supply feasible
    allocation of items in layer $\layer$. Furthermore, summing over
    all layers, the sum of jobs' valuations under this suite of
    allocations is exactly a quarter of that under $x$, whereas the
    cost is exactly a quarter of the cost
    $\sum_{r\geq 1}b'_{tr}c_{tr}$ defined above for the unscaled
    allocations $x^{(r)}$. We therefore conclude that
    \[\sum_r \fval(\xr) = \frac 14\sum_r \left(\sum_{j\in S^{(r)}} v_jx_j - b'_{tr}c_{tr}\right) \ge \frac{1}{4} \fval(x).\]

    We will now convert each $\xr$ into a unit allocation
    corresponding to layer $\layer$ by reducing it to an appropriate
    unit-capacity setting without costs. Observe that we can write
    \[\sum_r \fval(\xr) = \frac 14\sum_r \left(\sum_{j\in S^{(r)}} v_jx_j - b'_{tr}c_{tr}\right) = \frac{1}{4}\sum_{j\in S^{(r)}} \Big(v_j-\sum_{t\in I_j}c_{tr}\Big)x_j.\]

    Define a unit-capacity setting over the set of items $\layer$,
    where the value of a job $j$ is given by
    $v'_j=v_j-\sum_{t\in I_j}c_{tr}$. The allocation $\xr$ is
    demand-feasible for this setting. We can therefore apply
    Theorem~\ref{thm:unit-capacity} to obtain a unit allocation $\xpr$
    for this setting with the property that
    $\fval(\xr)\le O(\log L/\log\log L) \fval(\xpr)$. Let
    $\{T_{1r}, T_{2r}, \cdots\}$ denote the partition of items in
    $\layer$ corresponding to this unit allocation. Set
    $\tau_{kr} = \{(t,r): t\in T_{kr}\}$. Then the pair
    $(\sum_r \xpr,\{\tau_{k,r}\})$ forms a unit allocation for our
    original setting that obtains the claimed welfare bound.
\end{proof}

\section*{Acknowledgements}
We are grateful to Dimitris Paparas for discussions on this work.

% Bibliography
\bibliographystyle{plainnat}
\bibliography{newref,welfare,scheduling,agt,bmd,ea}

%\newpage
%\appendix

\end{document}